\documentclass[aps,pra,reprint,showpacs]{revtex4-1}

\usepackage{graphicx}
\usepackage{braket}
\usepackage{bm}
\usepackage{amsmath}
\usepackage{amsthm}
\usepackage{amssymb}
\usepackage{multirow}

\theoremstyle{definition}
\newtheorem{theorem}{Theorem}
\newtheorem*{remark}{Remark}
\DeclareMathOperator{\arccot}{arccot}

\begin{document}


\title{Measurement-assisted Landau-Zener transitions} 



\author{Alexander Pechen}
\email{pechen@mi.ras.ru}
\homepage{http://www.mathnet.ru/eng/person17991}
\affiliation{Steklov Mathematical Institute of Russian
Academy of Sciences, Gubkina 8, Moscow 119991, Russia}

\author{Anton Trushechkin}
\email{trushechkin@mi.ras.ru}
\homepage{http://www.mathnet.ru/eng/person31114}
\affiliation{Steklov Mathematical Institute of Russian
Academy of Sciences, Gubkina 8, Moscow 119991, Russia}
\affiliation{National Research Nuclear University ``MEPhI'', Kashirskoe shosse 31, Moscow 115409, Russia}


\date{\today}

\begin{abstract}
Nonselective quantum measurements, i.e., measurements without reading the results, are often considered as a resource for manipulating quantum systems. In this work, we investigate optimal acceleration of the Landau-Zener (LZ) transitions by non-selective quantum measurements. We use the measurements of a population of a diabatic state of the LZ system at certain time instants as control and find the optimal time instants which maximize the LZ transition. We find surprising nonmonotonic behavior of the maximal transition probability with increase of the coupling parameter when the number of measurements is large. This transition probability gives an optimal approximation to the fundamental quantum Zeno effect (which corresponds to continuous measurements) by a fixed number of discrete measurements. The difficulty for the analysis is that the transition probability as a function of time instants has a huge number of local maxima. We resolve this problem both analytically by asymptotic analysis and numerically by the development of efficient algorithms mainly based on the dynamic programming. The proposed numerical methods can be applied, besides this problem, to a wide class of measurement-based optimal control problems.
\end{abstract}


\maketitle 

\section{Introduction}

The control of atomic and molecular systems with quantum dynamics is an important branch of modern science with multiple existing and prospective applications in physics, chemistry, and quantum technologies~\cite{WisemanBook2014,Lapert2010,Brif2012,Caruso2012,Leung2012,
Brif2010,Holevo2013,TannorBook2007,AlessandroBook2007,LetokhovBook2007,BrumerBook2003,FradkovBook2003,RiceBook2000,ButkovskiyBook1990}. Significant efforts are directed towards the development of efficient methods for finding optimal controls for quantum systems~\cite{Machnes2011,Kuprov2011,Baturina,Sanders2014,qsl,Hentschel2011}. In this work we approach this problem for optimal acceleration of transitions in the two-level Landau-Zener system.

The Landau-Zener (LZ) system~\cite{Landau,Zener,Stuck,Majorana} is a simple model of nonadiabatic transition caused by avoided energy level crossing. This system has a very wide range of applications in physics, chemistry, and biochemistry; see review in~\cite{Pokrovsky}. Some of the applications include transfer of charge~\cite{Kuztetsov}, photosynthesis~\cite{Warshel}, chemical reactions~\cite{LZchem,Nitzan}, and manipulations with qubits. The last application includes the utilization of the Landau-Zener-St\"uckelberg interferometry~\cite{MachZenderLZ,SpinBeamspl,UltrafastLZS}, which is based on multiple traverses of the avoided crossing region, functional variation of time-dependent parameters of the model with a single traverse of the avoided crossing region~\cite{LZ-gate-robust,qsl,SupercondLZ,Qdriving,Baturina,TrapfreeLZ}, and coupling the LZ system to an external environment~\cite{LZ-spinchain,TempLZ,PhotonassistedLZ,LZ-mediated-envir}.

In this paper, we use another quantum control paradigm---measurement-based quantum control, which was proposed and developed in~\cite{PechenIlin,PechenChemPh,Shuang}. Experimental demonstrations of measurement-only state manipulation on a nuclear spin qubit in diamond by adaptive partial measurements is described in~\cite{Blok2014}. The use of measurements as control is now actively studied also in combination with feedback~\cite{Wiseman2011,Fu2014}. Machine learning was used to generate autonomous adaptive feedback schemes for quantum information~\cite{MLearn}. The measurement-based quantum control scheme~\cite{PechenIlin} utilizes non-selective quantum measurements, i.e., measurements without reading the results, to manipulate quantum systems. Such measurements do not increase knowledge about the quantum state of the system, but, in contrast to classical mechanics, they influence the system and can be regarded as ``kicks'' destroying quantum coherence. We consider the problem of maximization of the LZ transition probability from one adiabatic energy level to another using a fixed number of non-selective measurements of adiabatic level at certain time instants. The goal is to find optimal time instants to maximize the transition.

If the number of measurements is infinite and measurements are performed infinitely frequently, then the transition probability tends to one due to the fundamental quantum Zeno effect~\cite{Halfin,Halfin-2,SudarMisr}. However, an infinite number of measurement is never experimentally achievable. So, one may ask a natural question of how to find optimal instants of a given finite and fixed number of measurements to approximate the quantum Zeno effect. Approximaton to the quantum Zeno effect by finite number of measurements was considered in~\cite{Cook,Itano,Zeil-1,Zeil-2,Campisi2011} without formulating the problem of optimality. The problem of optimal approximation was stated and solved for two-level systems in~\cite{PechenIlin,Shuang} in the case when measured observables can be arbitrarily chosen. In the present work, for a particular model (the LZ model), we consider and solve analytically and numerically the more restrictive case when observables are fixed while the measurement's time instants are optimized. The advantage of this scheme is in its relative simplicity for experimental realization, since instants of measurements of a fixed observable (population) are simpler to modify than measured observables.

We find a surprising nonmonotonic dependence of the maximal transition probability on the coupling strength when the number of measurements is large. Recently a nonmonotonic dependence on the coupling strength and temperature was found for a different physical situation of a LZ system coupled to a harmonic-oscillator mode \cite{Ashhab2014}.

The following text is organized as follows. Section~\ref{SecLZmodel} is devoted to a discussion of the model and the LZ formula for the transition probability. In Section~\ref{SecLZControl} we discuss a measurement-based control scheme and formulate the optimal control problem. We show that the target function, which is the transition probability as a function of time instants, may have a huge number of local maxima. This feature makes it hard to solve the problem by well-known general global optimization methods such as random search, simulated annealing, evolutionary algorithms, etc.

In Section~\ref{SecLZdynprog}, we develop an efficient numerical algorithm based on dynamic programming to solve this problem. The algorithm is still computationally costly and in order to unravel mathematical structures standing behind this problem and to develop simpler algorithms, we consider limiting cases of antiadiabatic (Section~\ref{SecSmall}) and adiabatic (Section~\ref{SecOptLarge}) regimes. A graphical presentation of the results is given in Figs.~\ref{FigResults} and \ref{FigGraphn}.

\section{Landau--Zener model}\label{SecLZmodel}

The LZ Hamiltonian at time $t$ is defined as
\begin{equation}\label{EqLZHam}
H(t)=\Delta\sigma_x-\varepsilon t\sigma_z,
\end{equation}
where
\[\sigma_x=\begin{pmatrix}0&1\\1&0\end{pmatrix},\quad \sigma_z=\begin{pmatrix}1&0\\0&-1\end{pmatrix}\] are Pauli matrices, and $\varepsilon$ and $\Delta$ are positive constants. The time evolution of a pure state $|\psi(t)\rangle$ of the system is defined by the Schr\"odinger equation (in the units $\hbar=1$)
\[
 \frac{d|\psi(t)\rangle}{dt}=-iH(t)|\psi(t)\rangle
\]

Time-dependent eigenvalues (adiabatic energy levels) and eigenvectors of the Hamiltonian have the form
\begin{eqnarray*}
E_j(t)&=&\mp\sqrt{\Delta^2+(\varepsilon t)^2},\\
\ket{\varphi_j(t)}&=&a_j(t)\begin{pmatrix}
\varepsilon t\pm\sqrt{\Delta^2+(\varepsilon t)^2}\\-\Delta
\end{pmatrix}
\end{eqnarray*}
where $j=0,1$, the top sign in $\pm$ and $\mp$ corresponds to $j=0$, and the bottom sign corresponds to $j=1$, $a_j(t)\in\mathbb C$. If $a_j(t)$ are such that $\ket{\varphi_j(t)}$ are unit vectors for all $t$, then $\ket{\varphi_0(t)}$ and $\ket{\varphi_1(-t)}$ tend to $\ket0$ (up to a phase term) as $t\to+\infty$ and to $\ket1$ as $t\to-\infty$. Here,
$|0\rangle=\begin{pmatrix}1&0\end{pmatrix}^T$, $|1\rangle=\begin{pmatrix}0&1\end{pmatrix}^T$
is the standard basis of $\mathbb C^2$.

The dependence of energy levels on time is shown on Fig.~\ref{FigLevels}. A transition between the energy levels is possible mostly in the time interval when they are close to each other (so called avoided crossing region).

\begin{figure}
\includegraphics[width=\linewidth]{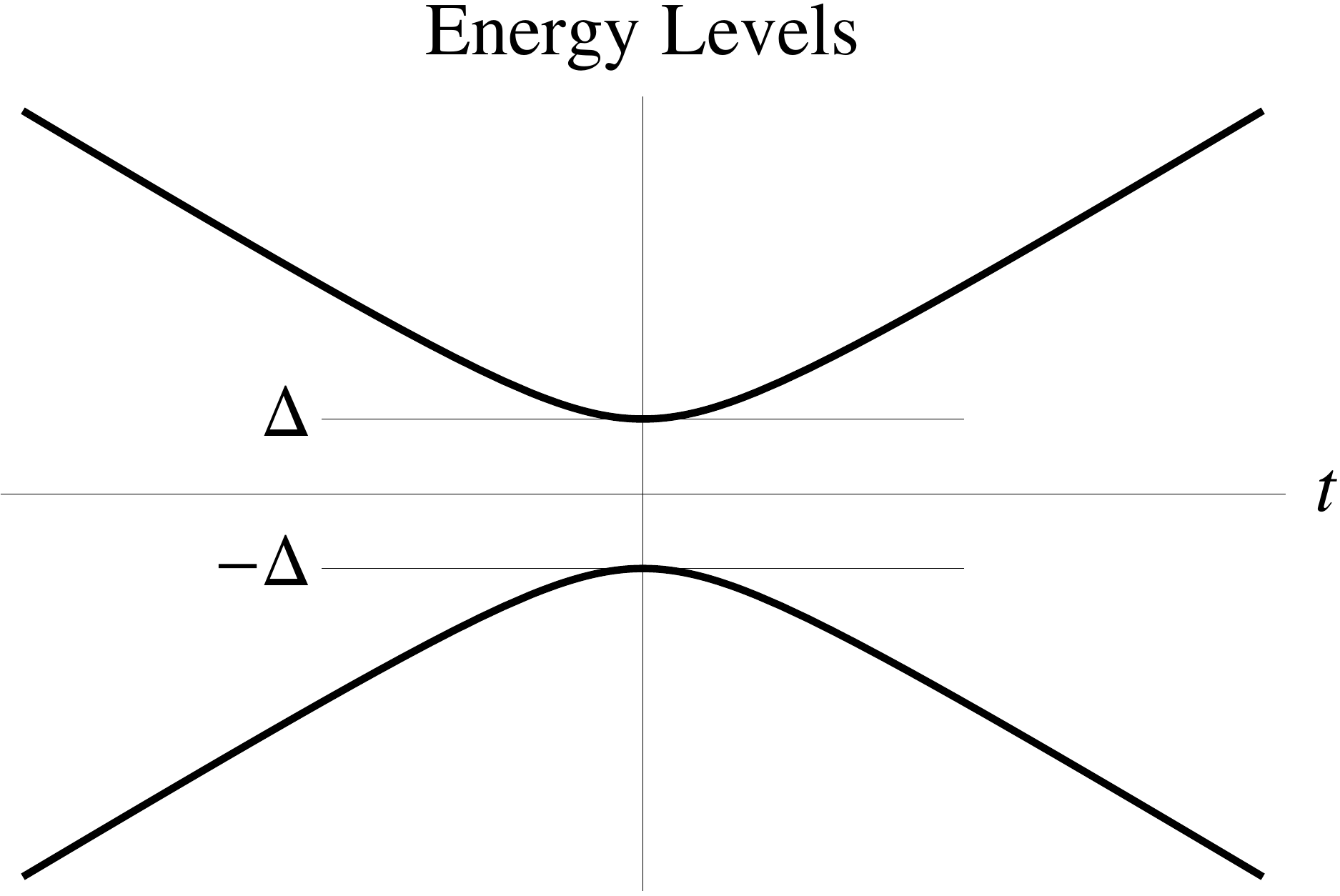}%
\caption{Energy levels of the LZ system\label{FigLevels}}%
\end{figure}

A pure quantum state is specified by a unit vector
\[
|\psi(t)\rangle=a(t)|0\rangle+b(t)|1\rangle,
\]
where $a(t)$ and $b(t)$ are complex coefficients, and $|a(t)|^2+|b(t)|^2=1$. Expressing these coefficients in the polar form as $a(t)=\cos(\theta/2) e^{i\varphi_1}$, $b(t)=\sin(\theta/2)e^{i\varphi_2}$, $0\leq\theta\leq\pi$, $0\leq\varphi_{1,2}<2\pi$, allows one to represent the state of the system as a point on the unit two-dimensional sphere, called the Bloch sphere, with spherical coordinates $\theta$ and $\varphi=\varphi_2-\varphi_1$.

The Schr\"odinger equation takes the form
\begin{equation*}
\left\lbrace
\begin{aligned}
i\dot a&=-\varepsilon ta+\Delta b\\
i\dot b&=\Delta a+\varepsilon tb.
\end{aligned}
\right.
\end{equation*}

By the time rescaling $\tau=\sqrt{2\varepsilon} t$ (and redenoting the new time variable $\tau$ again by $t$) this system is reduced to
\begin{equation}\label{EqLZ}
\left\lbrace
\begin{aligned}
i\dot a&=-\frac t2a+\sqrt\gamma b\\
i\dot b&=\sqrt\gamma a+\frac t2b,
\end{aligned}
\right.
\end{equation}
where $\gamma=\Delta^2/2\varepsilon$. This system of equations defines a family of unitary operators $U(t,t_0)$, $t_0,t\in\mathbb R$. Let us denote their matrix elements in the standard basis as
\[
\begin{pmatrix}
u_{00}(t,t_0) & u_{01}(t,t_0)\\
u_{10}(t,t_0) & u_{11}(t,t_0)
\end{pmatrix},
\]
where
\[u_{kl}(t,t_0)=\langle k|U(t,t_0)|l\rangle.\]
Then $u_{00}(t,t_0)=a(t)$ and $u_{10}(t,t_0)=b(t)$, where $a(t)$ and $b(t)$ are the solutions of system (\ref{EqLZ}) with the initial conditions $a(t_0)=1$, $b(t_0)=0$. Analogously, $u_{01}(t,t_0)=a(t)$ and $u_{11}(t,t_0)=b(t)$ if the initial conditions are $a(t_0)=0$, $b(t_0)=1$.

The general solution of system (\ref{EqLZ}) can be expressed as (detailed derivation can be found, e.g., in~\cite{BraatasRashba})
\begin{subequations}\label{EqLZGenSol}
\begin{eqnarray}
a(t)&=&c_1e^{-3i\pi/8}\sqrt\gamma D_{-i\gamma-1}(-e^{i\pi/4}t)\nonumber\\&&+
c_2e^{-3i\pi/8}D_{i\gamma}(e^{3i\pi/4}t)\\
b(t)&=&c_1e^{3i\pi/8} D^*_{i\gamma}(e^{3i\pi/4}t)\nonumber\\&&-
c_2e^{3i\pi/8}\sqrt\gamma  D^*_{-i\gamma-1}(-e^{i\pi/4}t).
\end{eqnarray}
\end{subequations}
Here $D_\nu(z)$ is the parabolic cylinder function defined for arbitrary complex order $\nu$ and argument $z$, $D^*_\nu(z)=D_{\bar \nu}(\bar z)$ is a complex conjugate function, and $c_1$ and $c_2$ are arbitrary constants. The parabolic cylinder functions are related to the Whittaker functions and the confluent hypergeometric functions; see~\cite{GR}.

We will use the following asymptotics of the functions $D_{i\gamma}(\pm e^{3i\pi/4}t)$ and $D_{-i\gamma-1}(\pm e^{i\pi/4}t)$ as $t\to+\infty$:
\begin{subequations}\label{EqDasymp}
\begin{eqnarray}
D_{i\gamma}(e^{3i\pi/4}t)&=& e^{-3\pi\gamma/4}e^{it^2/4}t^{i\gamma}+O(t^{-1}),\\
D_{i\gamma}(-e^{3i\pi/4}t)&=& e^{\pi\gamma/4}e^{it^2/4}t^{i\gamma}+O(t^{-2}),\\
D_{-i\gamma-1}(e^{i\pi/4}t)&=&O(t^{-1}),\\
D_{-i\gamma-1}(-e^{i\pi/4}t)&=&\frac{\sqrt{2\pi}e^{-\pi\gamma/4}}{\Gamma(1+i\gamma)}e^{it^2/4}t^{i\gamma}+O(t^{-1})\qquad
\end{eqnarray}
\end{subequations}
Let us note that both functions are bounded as $t\in\mathbb R$.

Consider the initial conditions such that
\begin{equation}\label{EqInitCond}
|a(-\infty)|=1,\quad b(-\infty)=0.
\end{equation}
Asymptotics~(\ref{EqDasymp}) imply that $c_1$ can be choosen as $e^{-\pi\gamma/4}e^{3i\pi/8}$ (this corresponds to the phase of $a(t)$ equal to $e^{it^2/4}|t|^{i\gamma}$ for large negative $t$) and $c_2=0$, so that
\begin{subequations}\label{EqaD}
\begin{eqnarray}
a(t)&=&e^{-\pi\gamma/4}D_{i\gamma}(e^{3i\pi/4}t),\\
b(t)&=&-e^{-\pi\gamma/4}e^{3i\pi/4}\sqrt\gamma  D^*_{-i\gamma-1}(-e^{i\pi/4}t).\label{EqbD}
\end{eqnarray}
\end{subequations}
This gives the celebrated Landau--Zener formula:
\begin{equation}\label{EqLZformula}
|a(+\infty)|^2=e^{-2\pi\gamma}
\end{equation}
The quantity $|a(+\infty)|^2$ is the probability that the final diabatic state of the system will be the same as the initial diabatic state. This probability will be of our interest in the following. Since it corresponds to a jump from one adiabatic energy level to another (the state $\ket0$ corresponds to large negative energy for $t\to-\infty$ and to large positive energy for $t\to-\infty$), we will refer to it as ``transition probability'' (sometimes the term ``LZ transition'' is used in the opposite sense: for a transition between diabatic states $\ket0$ and $\ket1$ and, so, staying on the same adiabatic energy level).

As we see, the transition probability is close to zero for large $\gamma$, or, equivalently, for small $\varepsilon$ with respect to $\Delta^2$. That is, for very slow (adiabatic) evolution of the Hamiltonian, the system stays on its initial adiabatic energy level. This limiting case is called adiabatic. In the opposite limit of small $\gamma$, called the anti-adiabatic limit, the LZ transition probability is close to one.

\section{Measurement-based quantum control}\label{SecLZControl}

A general state of a quantum system is specified by a density operator  $\rho$ which is a self-adjoint positive unit-trace operator. Its unitary evolution is given by
$$\rho(t)=U(t,t_0)\rho(t_0)U^\dag(t,t_0),$$
where unitary operators $U(t,t_0)$ are defined above.

A von Neumann observable for a two-level system is specified by two linear operators $P_0$ and $P_1$ such that $P_jP_k=P_j\delta_{jk}$, where $j,k=0,1$ and $\delta_{jk}$ is the Kronecker symbol, and $P_0+P_1=I$, where $I$ is the identity operator (operator $P$ satisfying the property $P^2=P$ is called projector).

We will use only non-selective measurements, i.e., measurements without reading the results. Such measurements do not increase our knowledge of the system's state, but they cause an instantaneous change of the state to
\begin{equation}\label{EqCollapseGen}
\rho'= P_0\rho P_0+P_1\rho P_1,
\end{equation}
where $\rho$ is the density matrix just before the measurement. In particular, if we perform a measurement in the standard basis, i.e., $P_0=\ket0\bra0$, $P_1=\ket1\bra1$, then
\begin{equation}\label{EqCollapse}
\rho'= |0\rangle\langle0|\rho|0\rangle\langle0|+|1\rangle\langle1|\rho|1\rangle\langle1|.
\end{equation}
Alternatively, one can represent the change of the quantum state as follows. Every two-dimensional density matrix in the standard basis can be represented as
\[
\rho=\frac12\begin{pmatrix}
1+w_z & w_x-iw_y\\
w_x+iw_y & 1-w_z
\end{pmatrix},
\]
where $\bm w=(w_x,w_y,w_z)\in\mathbb R^3$ is called the Bloch vector, $|\bm w|\leq1$. Measurement in the standard basis changes the state to
\[
\rho'=\frac12\begin{pmatrix}
1+w_z & 0\\
0 & 1-w_z
\end{pmatrix}.
\]
So, it eliminates the off-diagonal elements of the density matrix, or, equivalently, projects the Bloch vector onto the vertical axis. Unitary evolution does not change the length of the Bloch vector.

In the density matrix formalism,  initial conditions (\ref{EqInitCond}) have the form $\rho(-\infty)=|0\rangle\langle0|$. The LZ transition probability, i.e., the left-hand side of (\ref{EqLZformula}), is now expressed as $\langle0|\rho(+\infty)|0\rangle$.

We consider the following optimization problem: given natural $N$, find instants of measurements $t_1\leq t_2\leq\ldots\leq t_N$ such that the transition probability $\langle0|\rho(+\infty)|0\rangle$ is maximal for the initial condition $\rho(-\infty)=|0\rangle\langle0|$. The system evolves according to the unitary evolution between measurements and with jumps according to formula~(\ref{EqCollapse}) at the instants of measurements.

The transition probability can be expressed as
\begin{equation}\label{EqTarget}
\langle 0|\rho(+\infty)|0\rangle =\sum_{j_1,\ldots,j_N\in\{0,1\}}\prod_{k=0}^{N}|u_{j_{k+1},j_k}(t_k,t_{k+1})|^2,
\end{equation}
where $t_0=-\infty$, $t_{N+1}=+\infty$, and $j_0=j_{N+1}=0$.

We could use (\ref{EqLZGenSol}) to express the matrix elements $u_{jk}(t,t_0)$ by the parabolic cylinder functions and to solve the maximization problem numerically. However,  the matrix elements are high-oscillating functions with respect to both $t_0$ and $t$. Moreover, the frequency of oscillations tends to infinity as $\gamma\to\infty$. This causes a huge number of local maxima in the corresponding control landscape. As examples, target functions for some cases are presented on Figs.~\ref{FigTarget1dim} and \ref{FigTarget2dim}. Our attempts to apply the known methods of global optimization such as random search, simulated annealing, evolutionary algorithms, etc., show that the problem is hard to solve by these methods.

\begin{figure*}
\includegraphics[width=\linewidth]{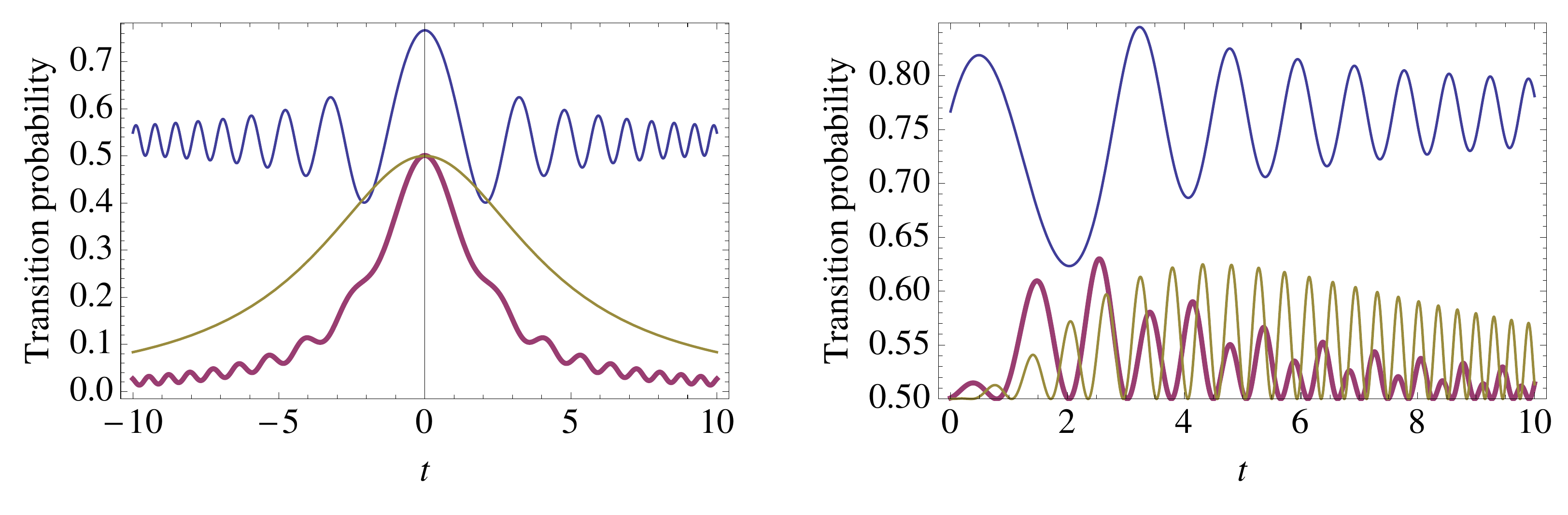}

\caption{Graphs of the transition probability for one measurement at the instant $t$ (left column) and for three measurements (right column) at the instants $-t$, 0, and $t$, $t\geq0$, for $\gamma=0.1$ (top curves), $\gamma=1$ (bottom thick curves), and $\gamma=5$ (bottom thin curves).\label{FigTarget1dim}}

\end{figure*}

\begin{figure*}
\includegraphics[width=\linewidth]{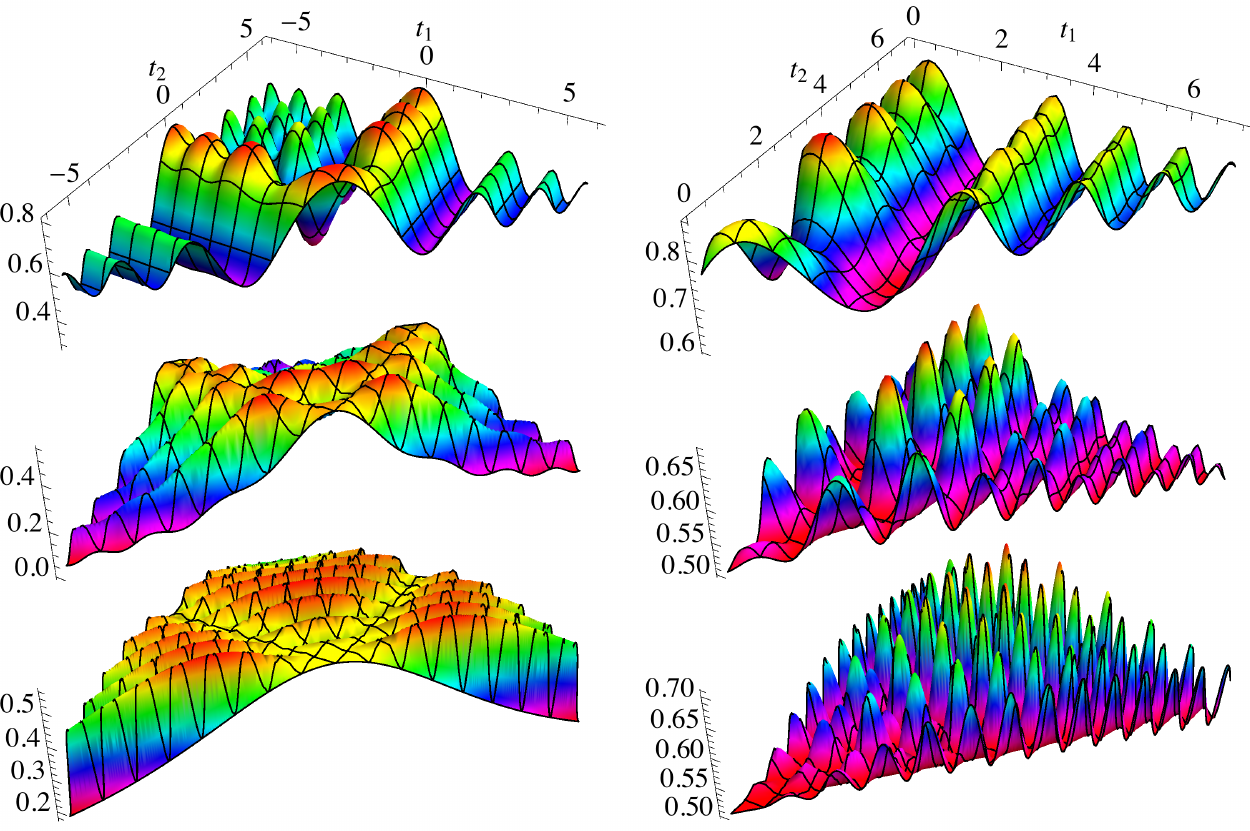}

\caption{Graphs of the transition probability for two measurements at the instants $t_1,t_2$ (left column) and for five measurements (right column) at the instants $-t_2,-t_1$, 0, $t_1$, $t_2$, $0\leq t_1\leq t_2$, for $\gamma=0.1$ (top plots), $\gamma=1$ (middle plots), and $\gamma=5$ (bottom plots). The ranges of $t_1$ and $t_2$ for the middle and bottom graphs are the same as for the top graphs. \label{FigTarget2dim}}

\end{figure*}

First let us establish an upper bound for the maximal transition probability.

\begin{theorem}\label{ThPechen}
For all values of the arguments we have
\begin{equation}\label{EqUpBnd}
\braket{0|\rho(+\infty)|0}\leq \frac12\left\lbrace1+\left(\cos\frac{\Delta\varphi}{N+1}\right)^{N+1}\right\rbrace.
\end{equation}
Here $\Delta\varphi$ denotes the angle between the Bloch vectors that correspond to the states $\ket0$ and $U^\dag(+\infty,-\infty)\ket0$.
\end{theorem}
\begin{proof}
Suppose we can arbitrarily choose not only  time instants of measurements but also observables. Then the state of the system after the measurement is changed according to~(\ref{EqCollapseGen}), where an observable $\{P_0,P_1\}$ is chosen by us and may be different for different measurements. Consider maximization of $\braket{0|\rho(+\infty)|0}$ under these conditions. This problem was solved in~\cite{PechenIlin,Shuang}; the maximal transition probability is given by the right-hand side of~(\ref{EqUpBnd}). Our problem has an additional restriction that the observables are fixed. The optimal value of the target function in the problem with an additional restriction cannot exceed the optimal value of the target function for the problem without this restriction.
\end{proof}

\begin{remark}
Formally, in the problem considered in~\cite{PechenIlin,Shuang}, the instants of measurements $t_1,\ldots,t_N$ are fixed. But, in the case of variable observables, this is not a restriction, because the measurement of any observable $\{P_1,P_2\}$ at an instant $t$ is equivalent to the measurement of the observable $\{P'_1,P'_2\}$ at the instant $t'$, where
$$P'_j=U(t',t)P_jU^\dag(t',t),\quad j=1,2.$$
\end{remark}

Further relations between these two measurement-based optimal control problems are considered in Sec.~\ref{SecPechen}.

Fig.~\ref{FigTarget1dim} suggests that, for all $\gamma$, the optimal instant of the measurement in the case of $N=1$ is $t=0$. We prove this rigorously for the limiting cases of small and large $\gamma$. The corresponding value of the transition probability equals to
\begin{equation}\label{EqTarget1}
\braket{0|\rho(+\infty)|0}=\frac12(1+e^{-2\pi\gamma}).
\end{equation}
This formula can be derived from
\begin{equation}\label{EqD0}
D_\nu(0)=\frac{2^{\nu/2}\sqrt\pi}{\Gamma(\frac{1-\nu}2)},
\end{equation}
where $\Gamma(x)$ is the gamma function, which immediately follows from the definition of the parabolic cylinder function and the confluent hypergeometric function, see~\cite{GR}.

This gives an important concluson that a single optimal measurement decreases the difference between the transition probability and unity twice. In particular, for large $\gamma$, the LZ transition probability exponentially tends to zero, while just one optimal measurement increases it to $1/2$.

\section{Dynamic programming algorithm}\label{SecLZdynprog}

Here we solve the optimal control problem using the dynamic programming paradigm~\cite{Taha}. Denote by $f_n(p,t)$ the maximal value of $\braket{0|\rho(+\infty)|0}$ provided that the time $t$ state of the system is
\begin{equation}\label{EqRho}
\rho(t)=p\ket0\bra0+(1-p)\ket1\bra1
\end{equation}
and $n$ measurements are allowed to perform after the time instant $t$. If $n=0$, then no more measurements are allowed and we have only unitary evolution from $t$ to $+\infty$ which gives
\[
f_0(p,t)=p\,|u_{00}(+\infty,t)|^2+(1-p)|u_{01}(+\infty,t)|^2.
\]
For $n\geq 1$, we have the following recurrent relation:
\begin{equation}\label{EqRecur}
f_n(p,t)=\max_{t'\geq t}f_{n-1}(p\,|u_{00}(t',t)|^2+(1-p)|u_{01}(t',t)|^2,t').
\end{equation}
The target function is $\braket{0|\rho(+\infty)|0}=f_N(0,-\infty)$. Thus, the problem of multidimensional maximization is reduced to a sequence of problems of one-dimensional maximization.

For further simplification, note that $f_n(p,t)=f_n(p',t)$ if $p-1/2$ and $p'-1/2$ have the same sign for any $n\geq0$, so that only whether $p<1/2$ or $p>1/2$ matters. Indeed,  the Bloch vector corresponding to a density operator of form (\ref{EqRho}) has the length $|2p-1|$ and is directed upwards if $p>1/2$ and downwards if $p<1/2$. The target function is $\braket{0|\rho(+\infty)|0}$, so, the Bloch vector corresponding to the final state $\rho(+\infty)$ should be as close as possible to the Bloch vector $(0,0,1)$, which corresponds to the state $\ket0\bra0$. The evolution of a Bloch vector under the action of both unitary evolution and non-selective measurements is uniform: if a Bloch vector $\bm w(t)$ evolves into $\bm w(t')$, then the Bloch vector $\lambda \bm w(t)$ evolves into $\lambda \bm w(t')$, $0\leq\lambda\leq|\bm w|^{-1}$. Hence, the optimal control depends only on the direction of the initial Bloch vector, but not on its length.

However, the optimal value of the target function depends also on the length of the initial Bloch vector. If $p=1/2$, then the quantum state is $\rho=I/2$ (sometimes it is called the chaotic state) and, by induction, one can show that $f_n(1/2,t)=1/2$ for all $n$ and $t$. As $p$ increases (decreases) from $1/2$ to $1$ ($0$), the function $f_n(p,t)$ changes linearly from $1/2$ to $f_n(1,t)$ [$f_n(0,t)$]. So,
\[
f_n(p,t)=\begin{cases}
(1-2p)[f_n(0,t)-1/2]+1/2,&p\leq1/2,\\
(2p-1)[f_n(1,t)-1/2]+1/2,&p\geq1/2,
\end{cases}
\]
and it is sufficient to solve the maximization problem in (\ref{EqRecur})  only for $p=0$ and $p=1$. For a fast and precise solution to this problem we use the so called very fast simulated annealing algorithm.

Now we can solve our maximization problem as follows. Successively, we calculate functions $f_n(p,t)$, $p=0,1$, $n=1,2,\ldots,N$ for $t\in[-50,50]$ with the step 0.01. This is sufficient for $\gamma\leq5$ and $N\leq15$, i.e., all optimal instants of measurements certainly lie within this interval. The results of the calculations are presented in Fig.~\ref{FigResults} and \ref{FigGraphn}.

\begin{figure*}
\includegraphics[width=\linewidth]{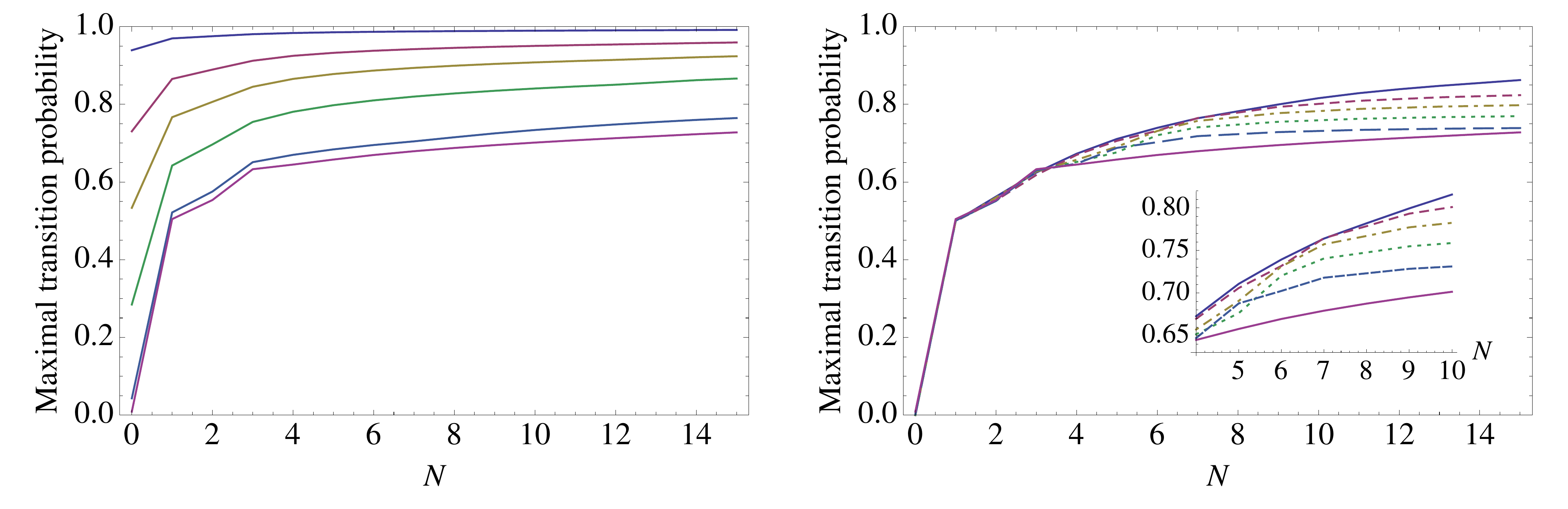}
\caption{Maximal transition probability vs number of measurements $N$ for different $\gamma$. Left (from top to bottom): $\gamma=0.01, \: 0.05,\: 0.1,\: 0.2,\: 0.5,\: 0.75$. Right (from bottom to top): $\gamma=0.75,\: 1,\: 1.2,\: 1.5,\: 2,\: 5$. Inset: a part of the same graph. Results for $\gamma\leq0.5$ were obtained using the algorithm for small $\gamma$, see Sec.~\ref{SecSmall}, and verified by a more precise and general, but more computationally costly method of Sec.~\ref{SecLZdynprog}. Differences in the maximal transition probabilities found by these two algorithms are less than $0.01$. Results for $\gamma>0.5$ were obtained using the general method of Sec.~\ref{SecLZdynprog}, but the differential evolution algorithm combined with the analysis of Sec.~\ref{SecOptLarge} also gives good results. The difference in the maximal transition probabilities is less than $0.01$ for $\gamma\geq1$ and less than or of the order of $0.01$ for $0.5\leq\gamma<1$. The curve for $\gamma=5$, as well as for larger $\gamma$ not shown here, almost coincides with the theoretical upper bound curve given by (\ref{EqTargetLargeOpt}). \label{FigResults}}
\end{figure*}

\begin{figure}
\includegraphics[width=\linewidth]{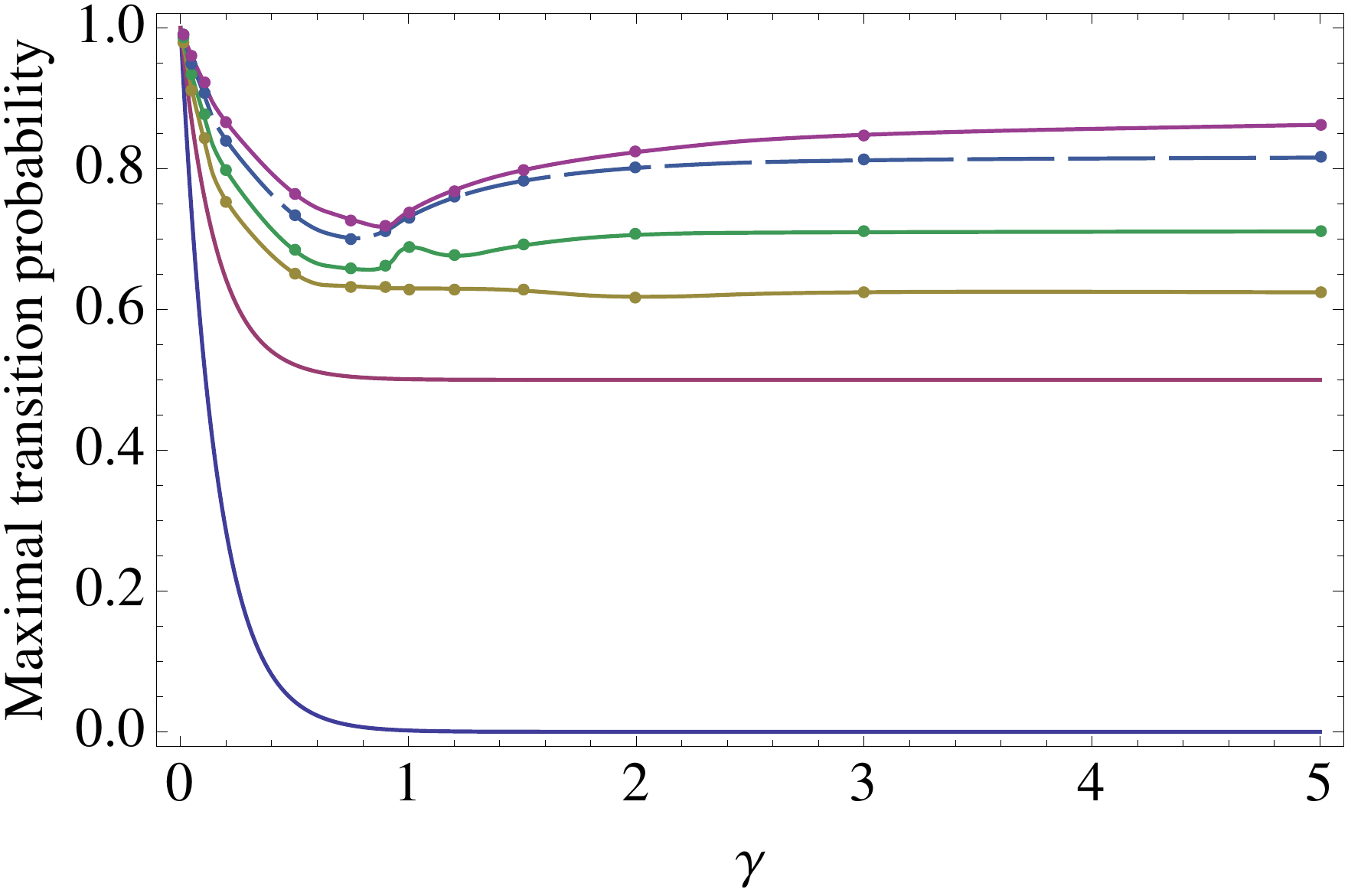}
\caption{Maximal transition probability vs $\gamma$ for a fixed number of measurements; from bottom to top $N=0, \:1,\: 3,\: 5,\: 10,\: 15$. Curves for $N=0$ and $N=1$ are drawn using the exact LZ formula~(\ref{EqLZformula}) and formula~(\ref{EqTarget1}). Other calculations are the same as described in the caption of Fig.~\ref{FigResults}.
\label{FigGraphn}}
\end{figure}

This algorithm is still computationally costly. In the next two sections we develop simpler algorithms and recover mathematical structures standing behind this problem. Toward this aim, we analyze the limiting cases of small and large $\gamma$ corresponding to antiadiabatic and adiabatic regimes.

\section{Anti-adiabatic regime}\label{SecSmall}

The anti-adiabatic regime is the case of small $\gamma$. It turns out (see the derivation in Appendix~\ref{App1}) that, in the first-order approximation with respect to $\gamma$, we have

\begin{widetext}
\begin{eqnarray}
|u_{00}(t,-\infty)|^2&=&1-\frac{\pi\gamma}2- \pi\gamma\sqrt2\,\Re \{e^{i\pi/4}\mathcal F^*(t/\sqrt2)\}-\gamma t^2\,\Re\, {}_2F_2(-it^2/2)+O(\gamma^2),\nonumber\\
|u_{00}(t,t_0)|^2&=&1+\Re\lbrace2\pi\gamma\mathcal F(t_0/\sqrt2)[\mathcal F^*(t/\sqrt2)-\mathcal F^*(t_0/\sqrt2)]-\gamma t^2\,{}_2F_2(-it^2/2)+\gamma t_0^2\,{}_2F_2(-it_0^2/2)\rbrace+O(\gamma^2),\nonumber\\
|u_{00}(+\infty,t_0)|^2&=&1-\frac{\pi\gamma}2+\Re\Bigl\{2\pi\gamma\mathcal F(t_0/\sqrt2)\Bigl[\frac{e^{-i\pi/4}}{\sqrt2}-\mathcal F^*(t_0/\sqrt2)\Bigr]+\gamma t_0^2\,{}_2F_2(-it_0^2/2)\Bigr\}+O(\gamma^2),\nonumber\\
\langle0|\rho(+\infty)|0\rangle&=&1-2\pi\gamma\sum_{k=0}^{N}\Re\{\mathcal F(t_k/\sqrt2)[\mathcal F^*(t_{k}/\sqrt2)-\mathcal F^*(t_{k+1}/\sqrt2)]\}+O(\gamma^2),\label{EqTargetsmall}
\end{eqnarray}
\end{widetext}
where $t_0=-\infty$, $t_{N+1}=+\infty$, and $\Re$ stands for the real part.

As we can see, we should keep only the terms with $j_1=j_2=\ldots=j_N=0$ in sum~(\ref{EqTarget}). Indeed, as we can see from~(\ref{Equutilde}) and (\ref{Equsmall}), all of the other terms have the order $\gamma^2$ or higher.

The value of the Fresnel integral $\mathcal F(t/\sqrt2)$ can be represented graphically by the Kornu spiral. Thus, from the geometric point of view, the problem in the first-order approximation is to find points $t_1,\ldots,t_N$ on the Kornu spiral which maximize function~(\ref{EqTargetsmall}) (see Fig.~\ref{FigKornu}).

Consider the case of a single measurement. Then the transition probability has the form
\[
\langle0|\rho(+\infty)|0\rangle=
1-\pi\gamma-2\pi\gamma|\mathcal F(t_1/\sqrt2)|^2+O(\gamma^2).
\]
Its maximum is achieved at $t_1=0$ and its value in the first-order approximation is $1-\pi\gamma$. Since the first-order approximation to the LZ formula gives the value $1-2\pi\gamma$, this finding coincides with a conclusion in the end of Sec.~\ref{SecLZControl} that a single measurement reduces the difference between one and the transition probability twice.

Taking the limit $N\to\infty$ in~(\ref{EqTargetsmall}) gives the value one for the transition probability, which corresponds to the quantum Zeno effect.

We cannot solve maximization problem (\ref{EqTargetsmall}) analytically for any $N$. This function also has many local maxima for direct numerical solution. Now we develop a numerical method for solving maximization problem (\ref{EqTargetsmall}) using the dynamic programming.

Let us define the functions (not to be confused with the functions $f_n$ in Sec.~\ref{SecLZdynprog})
\begin{eqnarray*}
&&f_0(t)=\Re\left\lbrace\mathcal F(t/\sqrt2)\left[\mathcal F^*(t/\sqrt2)-\frac{e^{-i\pi/4}}{\sqrt2}\right]\right\rbrace,\\
&&f_n(t)\!=\!\min_{t'\geq t}\left[\Re\left\lbrace\mathcal F(t/{\sqrt2})\left[\mathcal F^*({t}/{\sqrt2})-\mathcal F^*({t'}/{\sqrt2})\right]\right\rbrace \right.\\&&\qquad\,+ f_{n-1}(t')\Bigr],\\
&&f_N=\min_{t'}\left[\frac12\Re\left\lbrace e^{i\pi/4}\sqrt2\mathcal F^*(t'/\sqrt2)+1\right\rbrace+ f_{N-1}(t')\right],
\end{eqnarray*}
where $n=1,\ldots,N-1$. In the dynamic programming, one needs to consequently solve the minimization problems for $f_n(t)$, $n=1,\ldots,N$. The approximate value of the transition probability is then
\begin{equation}\label{EqDynProgTarg}
\langle0|\rho(+\infty)|0\rangle=1-2\pi\gamma f_N+O(\gamma^2).
\end{equation}
To solve the problem for arbitrary $N$, we successively calculate $f_n$, $n=1,\ldots,N$, for $t\in[-10,10]$ with the step $0.01$. This interval is large enough for $N\leq15$ since optimal time instants certainly lie within this interval.

To improve the accuracy, we perform the following procedure. Denote the solution found by the dynamic programming algorithm by $t_1,\ldots,t_N$. Then use the point $(t_1,\ldots,t_N)$ in the $N$-dimensional real search space as the initial point to search for the nearest local maximum of~(\ref{EqTargetsmall}). This is much easier than the search for the global maximum. Denote the new point as $(t'_1,\ldots,t'_N)$.

These operations do not depend on $\gamma$. According to~(\ref{EqTargetsmall}), approximation of the target function depends on $\gamma$, but the optimal time instants in the first-order approximation do not. So, it is enough to perform these operations only once. The results are provided in Table~\ref{Tbl1}. Figures~\ref{Figf0} and~\ref{FigKornu} are drawn using the data of Table~\ref{Tbl1}.

\begin{table*}
\caption{Solution of the problem in the first-order approximation with respect to $\gamma$\label{Tbl1}}

\begin{tabular}{c|c|l}
$N$ & $f_0$ & Optimal time instants\\\hline

1 & 0 & $t_1=0$\\\hline

2 & $0.188$ & $t_1=-3.31, t_2=0.12$\\\hline

3 & $0.360$ & $t_1=-3.33, t_2=0, t_3=3.33$\\\hline

4 &  $0.461$ & $t_1=-3.38, t_2=-0.24, t_3=0.24, t_4=3.38$\\\hline

5 & $0.521$ & $t_1=-3.41, t_2=-0.39, t_3=0, t_4=0.39, t_5=3.41$ \\\hline

6 & $0.563$ & $t_1=-3.44, t_2=-0.50, t_3=-0.17, t_4=0.17, t_5=0.50$,  $t_6=3.44$\\\hline

7 & $0.594$ & $t_1=-3.46, t_2=-0.59, t_3=-0.30, t_4=0, t_5=0.30$,  $t_6=0.59, t_7=3.46$\\\hline

8 & $0.619$ & $t_1=-3.48, t_2=-0.66, t_3=-0.39, t_4=-0.13, t_5=0.13$, $t_6=0.39, t_7=0.66$, $t_8=3.48$\\\hline

9 & $0.640$ & $t_1=-3.50, t_2=-0.73, t_3=-0.48, t_4=-0.24, t_5=0$, $t_6=0.24, t_7=0.48$, $t_8=0.73, t_9=3.50$\\\hline

10 & $0.657$ & $t_1=-3.52, t_2=-0.78, t_3=-0.56, t_4=-0.34, t_5=-0.11$, $t_6=0.11$, $t_7=0.34,$ $t_8=0.56, t_9=0.78, t_{10}=3.52$\\\hline

\multirow{2}{*}{11} & \multirow{2}{*}{$0.672$} & $t_1=-3.54, t_2=-0.83, t_3=-0.62, t_4=-0.42, t_5=-0.21$, $t_6=0$, $t_7=0.21$,  $t_8=0.42, t_9=0.62, t_{10}=0.83$, \\ & & $t_{11}=3.54$\\\hline

\multirow{2}{*}{12} & \multirow{2}{*}{$0.685$} & $t_1=-3.55, t_2=-0.88, t_3=-0.68, t_4=-0.49, t_5=-0.29$,
$t_6=0.10$, $t_7=0.10$,  $t_8=0.29, t_9=0.49, t_{10}=0.68$, \\ & & $t_{11}=0.88, t_{12}=3.55$\\\hline

\multirow{2}{*}{13} & \multirow{2}{*}{$0.697$} & $t_1=-3.55, t_2=-0.88, t_3=-0.69, t_4=-0.49, t_5=-0.30$,
$t_6=-0.10$, $t_7=0.10$,  $t_8=0.28, t_9=0.48, t_{10}=0.67$, \\ & & $t_{11}=0.88, t_{12}=3.54, t_{13}=7.08$\\\hline

\multirow{2}{*}{14} & \multirow{2}{*}{$0.709$} & $t_1=-7.08, t_2=-3.55, t_3=-0.87, t_4=-0.68, t_5=-0.48$, $t_6=-0.29$,
$t_7=-0.10$, $t_8=0.10, t_9=0.29, t_{10}=0.48$, \\ & & $t_{11}=0.68, t_{12}=0.87, t_{13}=3.55, t_{14}=7.08$\\\hline

\multirow{2}{*}{15} & \multirow{2}{*}{$0.720$} & $t_1=-7.09, t_2=-3.56, t_3=-0.91, t_4=-0.73, t_5=-0.55$, $t_6=-0.36$,
$t_7=-0.18$,  $t_8=0, t_9=0.18, t_{10}=0.36$, \\ & & $t_{11}=0.55, t_{12}=0.73, t_{13}=0.91, t_{14}=3.56, t_{15}=7.09$\\\hline

\end{tabular}
\end{table*}

\begin{figure}

\includegraphics[width=\linewidth]{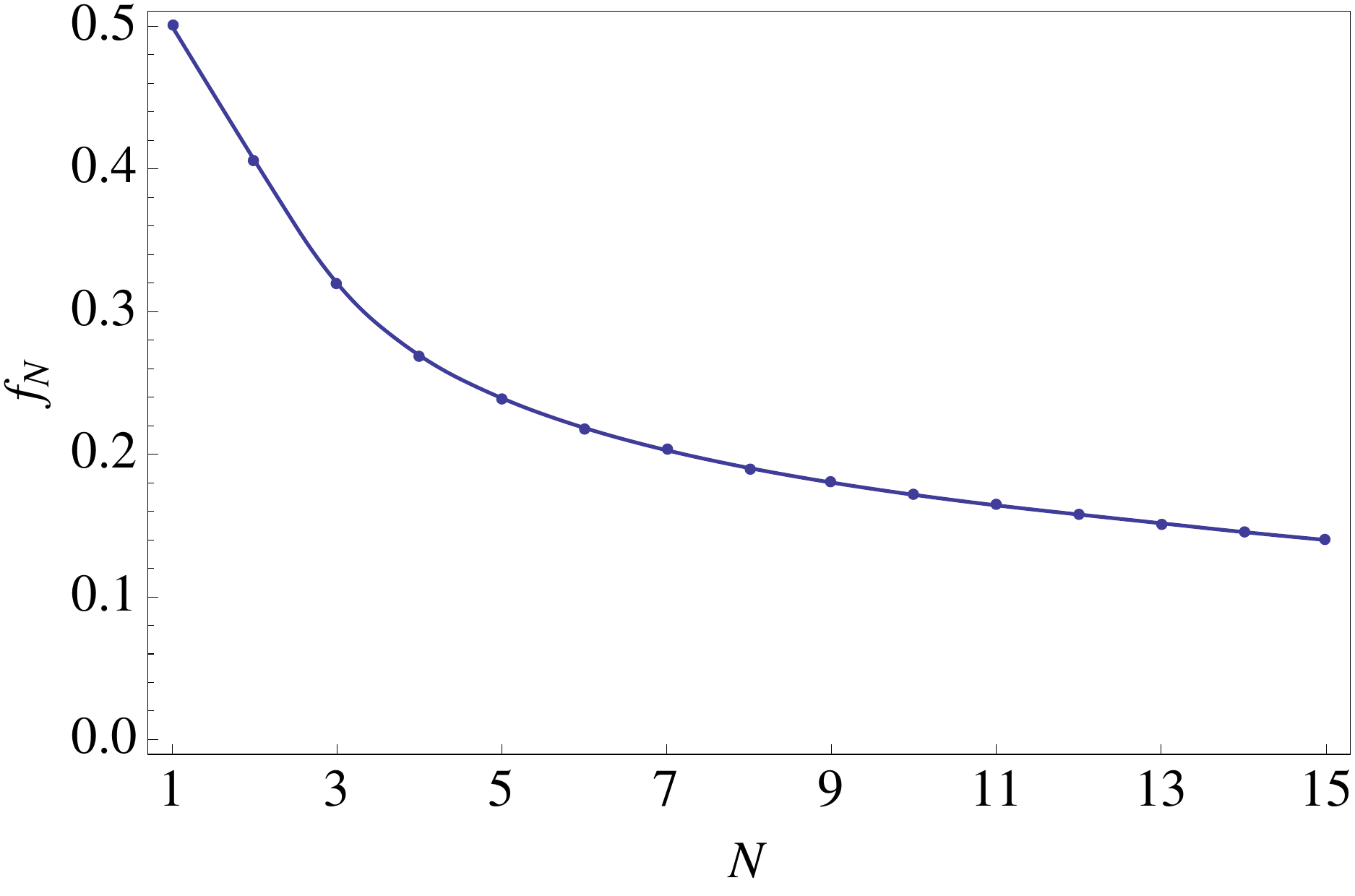}

\caption{$f_N$ vs number of measurements $N$. The target function is $\braket{0|\rho(+\infty)|0}=1-2\pi\gamma f_N+O(\gamma^2)$ [see (\ref{EqDynProgTarg})]. The case of no measurements corresponds to $f_0=1$ and the first-order approximation $1-2\pi\gamma$ of the LZ formula (for small $\gamma$, or anti-adiabatic regime). \label{Figf0}}
\end{figure}
\begin{figure}
\includegraphics[width=.8\linewidth]{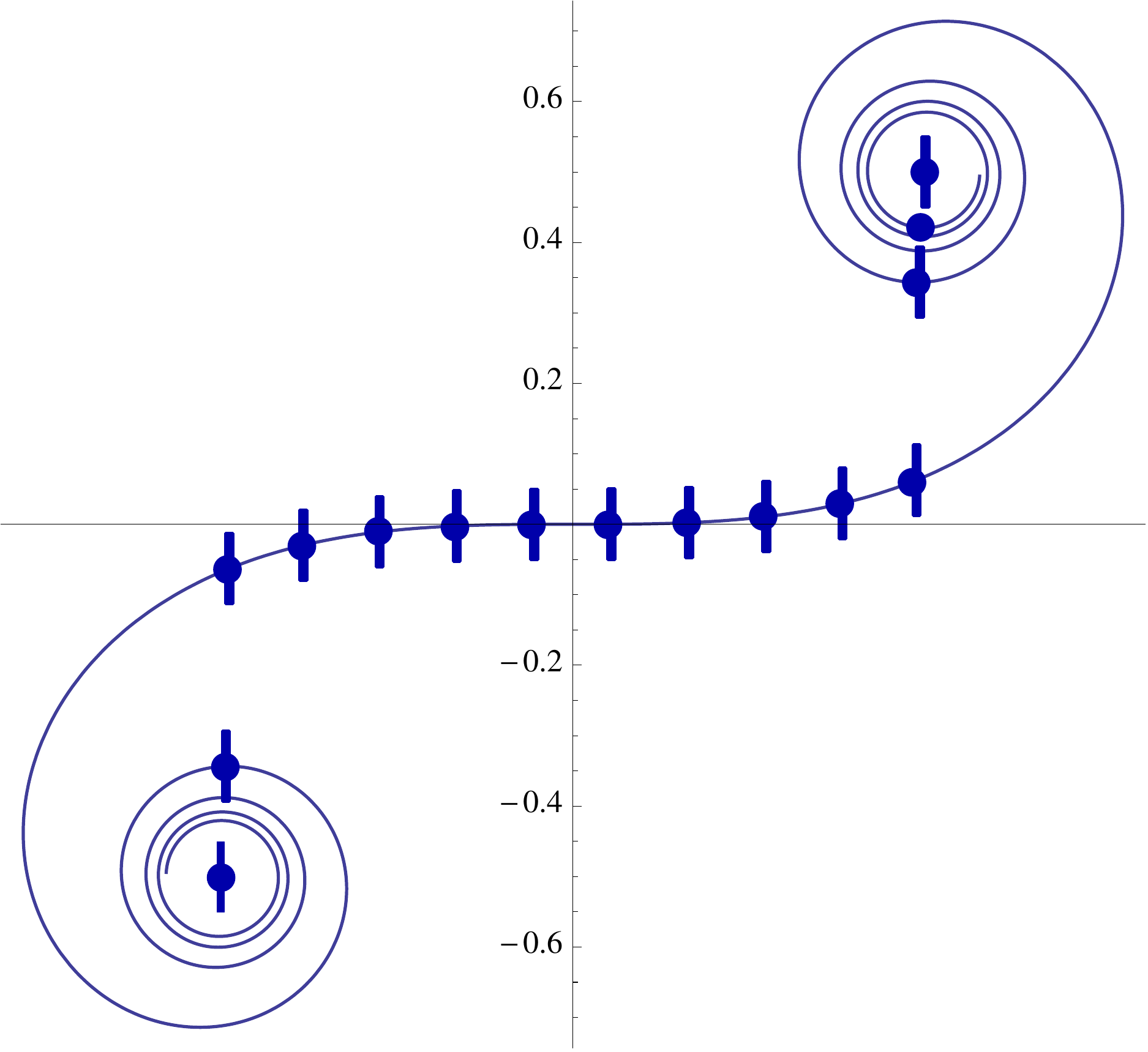}

\caption{Points on the Kornu spiral that correspond to the optimal instants of measurements for $N=12$ (sticks) and $N=13$ (balls), in the first-order approximation for small $\gamma$, or antiadiabatic regime. Points corresponding to $\pm\infty$ are also marked. Note that the optimal instants of measurements are symmetric with respect to zero for $N=12$ (this is the case for the majority of $N$s), but they are nonsymmetric for $N=13$. This ``symmetry breaking'' occurs when a new spire of the Kornu spiral becomes populated (see a point on a new spire in the top right corner of the graph) \label{FigKornu}}
\end{figure}

Our calculations show that these operations give very accurate results for $\gamma\le 0.1$; the value of the target function at the point $(t'_1,\ldots,t'_N)$ differs from the maximal value of the target function found by the more precise algorithm of Sec.~\ref{SecLZdynprog} by less than $0.01$.

For larger $\gamma$ we perform one more procedure. We use the final point $(t'_1,\ldots,t'_N)$ as an initial point to search for the nearest local maximum of the initial maximization problem~(\ref{EqTarget}) without assuming smallness of $\gamma$. This gives very accurate in the above sense results for $\gamma$ up to $0.5$. The results of this algorithm are presented on Fig.~\ref{FigResults}.

An interesting finding is that the obtained optimal instants are symmetric with respect to zero except of the cases $N=2$ and $N=13$. This ``symmetry breaking'' occurs when a new spire for the Kornu spiral becomes populated, see Fig.~\ref{FigKornu}. But due to the symmetry of the problem, if $t_1,\ldots,t_N$ is a solution of the optimization problem, then $-t_N,\ldots,-t_1$ is another solution which gives the same optimal value to the transition probability.

\section{Adiabatic regime}\label{SecOptLarge}

\subsection{Analytic solution}\label{SecLargeAnalyt}

The adiabatic regime corresponds to large $\gamma$. It turns out (see Appendix~\ref{App2} for details) that, in this case, the transition probability can be approximately expressed as
\begin{eqnarray}
&&\braket{0|\rho(+\infty)|0}=\frac12\lbrace 1-\cos\alpha_1\prod_{l=2}^{N}(\cos\alpha_l\cos\alpha_{l-1}\nonumber\\&&\qquad+
\sin\alpha_l\sin\alpha_{l-1}\cos\gamma\Delta\omega_l)\cos\alpha_N\rbrace,
\label{EqTargetLarge}
\end{eqnarray}
where $\alpha_k=\arccot(t_k/2\sqrt{\gamma})$, $k=1,\ldots,N$,
\begin{eqnarray*}
\omega(t)&=&\mathrm{sgn}(t)\ln(|t|+\sqrt{1+t^2})+t\sqrt{1+t^2},\\
\Delta\omega_k&=&\omega(t_k/2\sqrt\gamma)-\omega(t_{k-1}/2\sqrt\gamma).
\end{eqnarray*}

Let us consider the factors
\[
\cos\gamma\Delta\omega_l=\cos[\gamma(\omega_{t_l}-\omega_{t_{l-1}})].
\]
Since $\gamma$ is a large parameter, we can assign to these factors arbitrary values from $[-1,1]$ changing the time $t_l$ by an infinitesimal value of order $2\pi/\gamma\omega'(t_l)$. This can be understood also from Fig.~\ref{FigTarget1dim} (bottom right, $\gamma=5$, $N=3$) which shows rapid oscillations and top and bottom envelope curves; in the limit of infinitely rapid oscillations one can assign to the transition probability any value between the two envelope curves. Assigning the value $+1$ to all of these cosines gives
\[
\braket{0|\rho(+\infty)|0}\!=\!\frac12\left\lbrace 1\!-\!\cos\alpha_1\prod_{l=2}^{N}\cos(\alpha_{l-1}-\alpha_{l})\cos\alpha_N\right\rbrace.
\]
Put $\pi-\alpha_1=\varphi_0$, $\alpha_{l-1}-\alpha_{l}=\varphi_{l-1}$ ($l=2,\ldots,N$), $\alpha_N=\varphi_{N}$. Then
\[
\braket{0|\rho(+\infty)|0}=\frac12\left\lbrace 1+\prod_{l=0}^{N}\cos\varphi_l\right\rbrace,
\]
and the initial maximization problem is reduced to the problem
\[
\prod_{l=0}^{N}\cos\varphi_l\to\max_{\varphi_0,\ldots,\varphi_N}
\]
under the condition
\[
\sum_{l=0}^N\varphi_l=\pi.
\]
It can be shown \cite{Shuang} that the solution is $\varphi_l=\pi/(N+1)$ for all $l$. Hence, $\alpha_l=\pi-\pi l/(N+1)$, $\tau_l=-\cos[\pi l/(N+1)]$, $t_l=-2\sqrt{\gamma}\,\cot(\pi l/(N+1))$, $l=1,\ldots,N$. Thus, the optimal instants of measurements lie in neighborhoods of the instants
\begin{equation}\label{EqTimesOptLarge}
t_k=-2\sqrt\gamma\cot\left(\frac{\pi k}{N+1}\right)
\end{equation}
with radii of the order of $4\pi/\sqrt\gamma\omega'(t_k)$. The maximal value of the transition probability is
\begin{equation}\label{EqTargetLargeOpt}
\braket{0|\rho(+\infty)|0}=\frac12\left\lbrace1+\left(\cos\frac\pi{N+1}\right)^{N+1}\right\rbrace.
\end{equation}
Thus, a particular choice of the values of $\cos\gamma\Delta\omega_l$ gives a feasible solution. But, in view of Theorem~\ref{ThPechen} [with $\Delta\varphi=\pi$ in the right-hand side of (\ref{EqUpBnd}) since, up to a phase, $U^\dag(+\infty,-\infty)\ket0=\ket1$ for large $\gamma$], it is impossible to get a solution that exceeds the right-hand side of~(\ref{EqTargetLargeOpt}). Hence, solution (\ref{EqTimesOptLarge}) is optimal.

\begin{remark}
This optimal solution is not unique for $N\geq2$. In the case $N=2$,
\begin{eqnarray*}
&&\braket{0|\rho(+\infty)|0}\\&&=\frac12\left\lbrace1-\tau_1[\tau_1\tau_2+
\sqrt{(1-\tau_1^2)(1-\tau_2^2)}\cos\gamma\Delta\omega_2]\tau_2\right\rbrace
\end{eqnarray*}
($\tau_k=\cos\alpha_k$). If we, as above, put $\cos\gamma\Delta\omega_2=+1$, then the solution is
$\tau_1=-\cos(\pi/3)=-1/2$, $\tau_2=-\cos(2\pi/3)=1/2$.
However, if we put $\cos\gamma\Delta\omega_2=-1$, but change the sign of, for example, $\tau_1$, then, obviously, the value of the target function will be the same. So, in this case, $\tau_1=\tau_2=1/2$. However, these are different time instants with an infinitesimal distance from each other, such that $\cos\gamma\Delta\omega_2=-1$.
\end{remark}

So, the analytic solution gives only neighborhoods of the optimal time instants. Precise values can be determined numerically by finding the local maximum that is next to the point $(t_1,\ldots,t_N)$ given by  (\ref{EqTimesOptLarge}). Obviously, this problem is much simpler than a search for a global maximum in the initial problem.

However, the difficulty is that the size of the interval $4\pi/\sqrt\gamma\omega'(t_k)$ indefinitely increases as $t_k$ approaches zero, which happens for large $N$. If this size is comparable to or even larger than the differences $|t_k-t_{k\pm1}|$, then this approximation does not work. In this case, we propose the following approaches.

The first approach is to find optimal instants that are close to zero by another numerical global optimization method, e.g., random search or simulated annealing methods. This problem is low-dimensional and can be solved efficiently. For larger times, the approximation of large $\gamma$ can still be used. We developed such a generalization based on the results of Sec.~\ref{SecLargeAnalyt}.

The second approach is to use the differential evolution algorithm (a kind of genetic algorithm or, more generally, evolutionary algorithm) to find the maximum in the not infinitesimal vicinity of the point $(t_1,\ldots,t_N)$ given by  (\ref{EqTimesOptLarge}). Our calculation shows that this gives good results such that the differences between the maximal transition probabilities found by this method and by the more precise method of Sec~\ref{SecLZdynprog} are less than $0.01$ for $\gamma\geq1$ and less than or of the order of $0.01$ for $0.5\leq\gamma<1$. The use of the differential evolution method is much less computationally costly than the algorithm of Sec.~\ref{SecLZdynprog}. Recall that if $\gamma\leq0.5$, then the simple algorithm developed in Sec.~\ref{SecSmall} for the case of small $\gamma$ works well. Hence, in principle, for the whole range of $\gamma$, we can use rather fast algorithms of Sec.~\ref{SecSmall} and of this section, instead of the computationally costly algorithm of Sec.~\ref{SecLZdynprog}.

Let us stress that for the use of the differential evolution algorithm, it is important to restrict the search region to the vicinity of $(t_1,\ldots,t_N)$ given by~(\ref{EqTimesOptLarge}). If we search in the whole space $\mathbb R^N$, then the differential evolution as well as other algorithms of global optimization with the same computational power and time give large error: they find transition probabilities that are less than optimal by $0.1$ or more. So, the use of the analytic results of this section is important.

\subsection{Relation to measurement-based optimal control problem with variable obserables}\label{SecPechen}

In~\cite{PechenIlin,Shuang} the following measurement-based control problem for a two-level system is considered. Let there be no unitary evolution in the system; the initial state of the system be pure with the corresponding Bloch vector $\bm w_0$. We want to perform  $N$ non-selective measurements of arbitrary observables to maximize the probability of transition to another pure state with Bloch vector $\bm w_{\text{target}}$. The angle between $\bm w_0$ and $\bm w_{\text{target}}$ is denoted by $\Delta\varphi$. Let us describe the solution of this problem. Denote $\bm w_1,\ldots,\bm w_N$ as the vectors obtained by rotation of $\bm w_0$ by the angles $k\Delta\varphi/(N+1)$ in the plane formed by $\bm w_0$ and $\bm w_{\text{target}}$. Then the optimal observables are the projectors onto the states with the Bloch vectors $\bm w_1,\ldots,\bm w_N$. Denote them as $P_{\bm w_1},\ldots,P_{\bm w_N}$. The optimal value of the target function is
\[
\frac12\left\lbrace1+\left(\cos\frac{\Delta\varphi}{N+1}\right)^{N+1}\right\rbrace.
\]

This solution can be generalized to the case with unitary dynamics, which was also considered in~\cite{PechenIlin}. Let the system evolve according to the family of unitary operators $U(t,t_0)$. We want to find $N$ observables measured at fixed time instants $t_0\leq t_1\leq\ldots\leq t_N\leq T$ to maximize the transition probability to a desired state $\rho_{\text{target}}$ at a final time $T$. This problem is equivalent to the above problem without unitary evolution and with target state $U^\dag(T,t_0)\rho_{\text{target}} U(T,t_0)$. If $P_{\bm w_1},\ldots,P_{\bm w_N}$ are the optimal projectors in the evolution-free problem, then the optimal projectors for the problem with evolution are
$$U(t_k,t_0) P_{\bm w_k}U^\dag(t_k,t_0),\quad k=1,\ldots,N.$$

In our case, $t_0=-\infty$, $T=+\infty$, the initial state is $\rho(-\infty)=\ket0\bra0$, the target state is $\rho_{\text{target}}=\ket0\bra0$, and $U^\dag(+\infty,-\infty)\rho_{\text{target}}U(+\infty,-\infty)=\ket1\bra1$. Thus, $\Delta\varphi=\pi$. Further, in our problem we have the fixed observable $P_0=\ket0\bra0$, but we can choose the time instants $t_1,\ldots,t_N$. Hence, we must try to choose instants of measurements such that
\begin{equation}\label{EqIdealCond}
U(t_k,t_0) P_{\bm w_k}U^\dag(t_k,t_0)=P_0,\quad k=1,\ldots,N,
\end{equation}
or, equivalently,
$$U^\dag(t_k,t_0) P_0U(t_k,t_0)=P_{\bm w_k},\quad k=1,\ldots,N,$$
where $\bm w_k$ are defined as above. Of course, this is impossible in general. However, let us consider the trajectory of the Bloch vector corresponding to the state
$U^\dag(t,-\infty)\ket0$ for $-\infty\leq t\leq+\infty$ and large $\gamma$. The vector starts from the north pole as $t=-\infty$ and finishes in the south pole as $t=+\infty$. As we see from (\ref{EqLZlargegamma}) and (\ref{EqInitConstA}), its trajectory on the Bloch sphere is an arc, on which the Bloch vectors corresponding to the optimal projectors lie. This allows the target function to achieve the upper bound (\ref{EqTargetLargeOpt}) in the limit of large $\gamma$.

Thus, in this case, we can speak about a kind of duality of the problem with fixed instants of measurements and variable observables and the  problem with fixed observables and variable instants of measurements.

Fig.~\ref{FigGraphn} shows the dependence of the maximal transition probability on $\gamma$ for various $N$. We see that for a large number of measurements, the dependence is nonmonotonic (also we see it in Fig~\ref{FigResults}). The initial decrease of the transition probability with increasing $\gamma$ is in correspondence with the LZ formula. But, for $\gamma$ greater than some value (approximately from $0.7$ to $0.9$ depending on $N$), the probability increases up to the limit~(\ref{EqTargetLargeOpt}). This is caused by rapid oscillations, which allow one to choose the time instants such that the ``ideal'' conditions~(\ref{EqIdealCond}) are approximately satisfied with high precision.

\subsection{Maximin problem}\label{SecMaximin}

As shown above, the transition probability oscillates with high frequency for large $\gamma$. This is due to the factors $\cos(\gamma\Delta\omega_l)$ in~(\ref{EqTargetLarge}). A precise choice of time instants in neighborhoods with radii of orders $4\pi/\sqrt\gamma\omega'(t_k)$ allows one to assign definite values to these factors. However, these results can be used in practice only if the operation time of the measurement device is smaller than  $4\pi/\sqrt\gamma\omega'(t_k)$. Otherwise, if the actual instants of measurement can not be controlled with the necessary accuracy, then the final transition probability can be far from optimal due to oscillations. This can be seen in Fig.~\ref{FigTarget1dim} (bottom right) or Fig.~\ref{FigTarget2dim} (bottom right).

So, it makes it important to consider a maximin problem, that is, the problem of maximization of the transition probability with the worst instants of measurements within the specified neighborhoods. In other words, for given $\tau_1,\ldots,\tau_N$, the device ``chooses'' values of $\cos(\gamma\Delta\omega_l)$ that minimize the transition probability. The problem can be formulated as follows:
\begin{widetext}
\[
\braket{0|\rho(+\infty)|0}=\max_{\alpha_1\geq\ldots\geq\alpha_N}\min_{\Delta\omega_2,\ldots,\Delta\omega_N}\frac12\left\lbrace 1-\cos\alpha_1\prod_{l=2}^{N}(\cos\alpha_l\cos\alpha_{l-1}+
\sin\alpha_l\sin\alpha_{l-1}\cos\gamma\Delta\omega_l)\cos\alpha_N\right\rbrace.
\]
\end{widetext}

Let us analyze this problem. Since the terms $\cos(\gamma\Delta\omega_l)$ enter the target function linearly, the minimizing values are $\pm1$. Hence, the problem can be rewritten as
\begin{eqnarray}
&&\braket{0|\rho(+\infty)|0}=\label{EqMaximin}\\&&\max_{\alpha_1\geq\ldots\geq\alpha_N}\min_{\text{signs}}\frac12\left\lbrace 1-\cos\alpha_1\prod_{l=2}^{N}\cos(\alpha_{l-1}\pm\alpha_{l})\cos\alpha_N\right\rbrace\,.
\nonumber
\end{eqnarray}
Solving this problem is equivalent to solving
\begin{equation}\label{EqMinimax}
\min_{\alpha_1\geq\ldots\geq\alpha_N}\max_{\text{signs}}\cos\alpha_1
\prod_{l=2}^{N}\cos(\alpha_{l-1}\pm\alpha_{l})\cos\alpha_N.
\end{equation}
\begin{theorem}
All solutions of maximin problem (\ref{EqMaximin}) are as follows:\\
(1) $\alpha_1=\pi/2$, all other $\alpha_l$ are arbitrary not exceeding $\alpha_1$;\\
(2) $\alpha_N=\pi/2$, all other $\alpha_l$ are arbitrary not less than $\alpha_N$;\\
(3) every $\alpha_l$ is either 0, $\pi/2$, or $\pi$, and $\alpha_l=\pi/2$ at least for one $l$ and the order is satisfied: $\alpha_1\geq\alpha_2\geq\ldots\geq\alpha_N$.

In all cases, the transition probability is $\braket{0|\rho(+\infty)|0}=1/2$.
\end{theorem}
\begin{proof}
Obviously, the transition probability for all three cases (1)--(3) is equal to 1/2. Let us prove that these and only these arguments give the optimal solution. It is sufficient to show that for other values of $\alpha_l$, we can choose the signs in (\ref{EqMinimax}) such that the right-hand side of (\ref{EqMinimax})  is positive.

Let both $\alpha_1$ and $\alpha_N$ belong to either $[0,\pi/2)$ or $(\pi/2,\pi]$. Then, $\cos\alpha_1\cos\alpha_N>0$. Let us assign all signs to ``$-$''. All differences $\alpha_{l-1}-\alpha_{l}$ belong to the interval $(-\pi/2,\pi/2)$ and therefore $\cos(\alpha_{l-1}-\alpha_{l})>0$. Then the right-hand side of (\ref{EqMinimax}) is positive.

Now let $\alpha_1\in(\pi/2,\pi]$ and $\alpha_N\in[0,\pi/2)$. Then, $\cos\alpha_1\cos\alpha_N<0$.  If the solution is not of form 3), then there exist $\alpha_l\in(0,\pi/2)\cup(\pi/2,\pi)$. For definiteness, suppose that there exists at least one $\alpha_l\in(0,\pi/2)$.

We have $\alpha_1\geq\alpha_2\ldots\geq\alpha_N$. Let $k$ denote a number such that  $\alpha_k\in[\pi/2,\pi]$, $\alpha_{k+1}\in(0,\pi/2)$. Let us assign the signs for $l\leq k$ and $l\geq k+2$ to ``$-$'' and the sign for $l=k+1$ to ``$+$''. Then, $$\prod_{l=1}^k\cos(\alpha_{l-1}-\alpha_{l})>0,
\quad\prod_{l=k+2}^N\cos(\alpha_{l-1}-\alpha_{l})>0\,.$$
But $\alpha_{k}+\alpha_{k+1}\in(\pi/2,3\pi/2)$, so $\cos(\alpha_{k}+\alpha_{k+1})<0$. Combining these gives
\begin{eqnarray*}
&&\cos\alpha_1\prod_{l=1}^k\cos(\alpha_{l-1}-\alpha_{l})\cos(\alpha_{k}+\alpha_{k+1})\\
&&\qquad\times
\prod_{l=k+2}^N\cos(\alpha_{l-1}-\alpha_{l})\cos\alpha_N>0.
\end{eqnarray*}
\end{proof}

Thus, if the operation time of a measurement device is comparable or exceeds the period of oscillations $4\pi/\sqrt\gamma\omega'(t_k)$, then the optimal solution is a single measurement at the instant  $t=0$ which corresponds to $\alpha=\pi/2$; repeated measurements at this instant as well as measurements at  $t=\pm\infty$, which correspond to $\alpha=0$ and $\alpha=\pi$, do not affect the transition probability. In this case, the transition probability has the form
\[
\braket{0|\rho(+\infty)|0}=\frac{1-\tau_1^2}2,
\]
It has no oscillating terms which otherwise would require a precise choice of instants of measurements.

\section{Conclusions and discussion}

This work analyzes the measurement-assisted acceleration of transitions in the LZ system. Control by nonselective measurements is considered as an alternative tool to other methods of controlling the LZ system. The main numerical method exploited in this work is dynamic programming which seems to be natural for measurement-based optimal control problems.

We have obtained analytic behavior for small and large values of $\gamma$ and developed various numerical algorithms for solving the problem of measurement-assisted acceleration of the LZ transition. The combination of these methods allows one to solve the problem for the whole range of $\gamma$. The results are presented in Figs.~\ref{FigResults} and~\ref{FigGraphn}. The range $[0,5]$ for $\gamma$ on these figures covers values of $\gamma$ in experimental works, e.g., in~\cite{Qdriving,SupercondLZ,TunableLZ-BEC}. The increase of $\gamma$ above the value $\gamma=5$ almost does not change the results due to achievement of the upper bound~(\ref{EqTargetLargeOpt}). Hence the results are complete for practical applications.

We discover a surprising effect of non-monotonic dependence of the maximal transition probability on $\gamma$ for a large number of measurements. Decrease of the transition probability for small $\gamma$ is in agreement with the LZ formula, but for large $\gamma$, oscillation effects become essential and can be exploited to increase  the transition probability.

The obtained transition probability in the limit of an infinite number of measurements tends to one that recovers the quantum Zeno effect.  Surprisingly, as can be seen from Fig.~\ref{FigResults}, the convergence to this limit can be very slow for the LZ system, especially for values of $\gamma$ in the intermediate range 0.2--2.0.

The optimization techniques of this work can be applied to the measurement-assisted control of the dynamics of physical
quantum systems if the description by Hamiltonian~(\ref{EqLZHam}) is a suitable approximation. The main experimental challenge may be in performing non-selective measurements fast enough so that the duration of a measurement is much less than the intervals between the optimal time instants. The inability to perform measurements with good enough time precision can produce lower values of the observed transition probability. The impossibility of arbitrarily fast measurement was also argued to be a bottleneck for the quantum Zeno effect due to the time-energy uncertainty relation~\cite{Ghirardi1979}. The advantage of the methods of Sections~\ref{SecSmall} and~\ref{SecOptLarge} is in their robustness due to the final local search which can find optimal solutions for slightly distinct models, e.g., if measurement imperfections are small.

The analysis of Section~\ref{SecMaximin} also provides an answer for the adiabatic regime. In this regime, if time instants of measurements are not well defined and it significantly affects the desired transition probability, then the optimal solution will be to perform only one measurement at the avoided crossing time instant. This measurement will increase the transition probability from almost zero to 1/2, which agrees with the analysis of the quantum Zeno effect in the LZ system based on weak measurements~\cite{LZweak} and the analysis of the LZ dynamics under strong dephasing~\cite{Kayanuma-prl,Ao1989}. Generally, the formalism of continuous or weak measurements can be used if the duration of measurements is comparable to or larger than the intervals between the optimal time instants~\cite{Holevo,Peres,Ivanov,Lidar-prl,Lidar}.

If the description of the physical system requires small corrections to the Hamiltonian~(\ref{EqLZHam}), then the methods of Secs.~\ref{SecSmall} and~\ref{SecOptLarge} still can be applied due to the robustness of
the final local search step. Strong deviations which involve more general and complicated Hamiltonians with, for example, nonlinear dependence on time, coupling to an environment, presence of noise, etc., or even with not exactly known Hamiltonians, are beyond the scope of this work. However, the general approach of Sec.~\ref{SecLZdynprog} based on dynamic programming seems to be natural for measurement-based optimal control due to the discrete nature of control and we expect that it can be applied even to such general situations.

\begin{acknowledgments}
This work was supported by the Russian Science Foundation (Project No. 14-11-00687).
\end{acknowledgments}

\bigskip

\appendix

\section{Approximations of the unitary matrix elements for the anti-adiabatic regime}\label{App1}

The aim of this appendix is to obtain formulas~(\ref{EqTargetsmall}). At first, let us obtain approximations for matrix elements of operators $U(t,t_0)$ in the zeroth order approximation with respect to $\gamma$. Substitution of $\gamma=0$ into~(\ref{EqLZGenSol}) gives constant populations $|a(t)|^2$ and $|b(t)|^2$ [we can use formula $D_0(e^{3i\pi/4}t)=e^{it^2/4}$ or just solve system~(\ref{EqLZ}) with $\gamma=0$]. In particular, under the initial condition $|a(-\infty)|=1$, we have  $|a(+\infty)|=1$. This corresponds to the zeroth-order approximation to the LZ formula.

To obtain the next approximation, we use the following formula for the derivative of the parabolic cylinder function \cite{Brychkov}:
\begin{widetext}
\[
\left.\frac{\partial D_{i\gamma}(e^{3i\pi/4}t)}{\partial\gamma}\right|_{\gamma=0}=
\left\lbrace-\frac{\pi}{\sqrt2}e^{i\pi/4}\mathcal F^*(t/\sqrt2)-\frac{t^2}2\,{}_2F_2\left(-\frac{it^2}2\right)-\frac i2(\ln2+\bm C)\right\rbrace e^{it^2/4}.
\]
\end{widetext}
Here
\[
\mathcal F(t)=\sqrt{\frac2\pi}\int_0^te^{is^2}ds
\]
is the Fresnel integral,
\begin{equation}\label{EqGenHyperGeom}
{}_pF_q\left(\left.\begin{matrix}
a_1,\ldots,a_p\\b_1,\ldots,b_q
\end{matrix}\right|z\right)=\sum_{n=0}^\infty
\frac{(a_1)_n\cdots(a_p)_n}{(b_1)_n\cdots(b_q)_n}\frac{z^n}{n!}
\end{equation}
is the generalized hypergeometric function and $(a)_n=a(a+1)\cdots(a+n-1)$ is a Pochhammer symbol [$(a)_0=1$]. Often we will omit the parameters of the generalized hypergeometric function.

Since
\[
\sqrt\gamma D_{-i\gamma-1}(-e^{i\pi/4}t)=\sqrt\gamma D_{-1}(-e^{i\pi/4}t)+O(\gamma^{3/2}),
\]
we do not need the derivative of $D_{-i\gamma-1}(-e^{i\pi/4}t)$ if we are interested only in the first-order approximation. We have \cite{GR}
\[
D_{-1}(-e^{i\pi/4}t)=\sqrt{\frac\pi2}\,e^{it^2/4}\left[1+\sqrt2e^{i\pi/4}\mathcal F^*(t/\sqrt2)\right].
\]

Then the matrix elements $u_{j0}(t,t_0)$, $j=0,1$, can be written as
\begin{subequations}\label{Equutilde}
\begin{eqnarray}
u_{j0}(t,t_0)&=&\tilde u_{j0}(t,t_0)e^{i(t-t_0)^2/4}\left(\frac{1+t^2}{1+t_0^2}\right)^{i\gamma/2},\\
\tilde u_{00}(t,t_0)&=&1+\pi\gamma \mathcal F(t_0/\sqrt2)[\mathcal F^*(t/\sqrt2)-\mathcal F^*(t_0/\sqrt2)]\nonumber\\
&&-\frac{\gamma t^2}2\,{}_2F_2(-\frac{it^2}2)+\frac{\gamma t_0^2}2\,{}_2F_2(-\frac{it_0^2}2)\nonumber\\
&&-\frac{i\gamma}2\ln\frac{1+t^2}{1+t_0^2}+O(\gamma^2),\\
\tilde u_{10}(t,t_0)&=&-i\sqrt{\pi\gamma}[\mathcal F(t/\sqrt2)-\mathcal F(t_0/\sqrt2)]+O(\gamma^{\frac32})\qquad
\end{eqnarray}
\end{subequations}
Since equations (\ref{EqLZ}) transform one into another by complex conjugation and substitution $t\to-t$, we have
\[
u_{j1}(t,t_0)=u_{1-j,0}^*(-t,-t_0).
\]
So, we have calculated the first-order approximations of matrix elements of unitary evolution with respect to $\gamma$.

The separation of the term $[(1+t^2)/(1+t_0^2)]^{i\gamma/2}$ is motivated by unboundedness of the function ${}_2F_2(1,1;3/2,2;-it^2/2)$ as $t\to\pm\infty$, which we will see in a moment. Now we can explain it as follows. The term  $|t|^{i\gamma}$ in asymptotics~(\ref{EqDasymp}) is equal to one by the absolute value, but its first approximation with respect to  $\gamma$ is $1+i\gamma\ln|t|$. This is an unbounded function, which approximates $|t|^{i\gamma}$ incorrectly for large $t$. This approximation is adequate only for bounded $|t|$. We need to consider $t\to\pm\infty$ as well, and the LZ formula is formulated for this case, hence we must separate this term. Otherwise, as we will see, the first-order approximation of the unitary matrix elements will be invalid for large $|t|$ and we will not be able to reproduce the first-order approximation even for the case of no measurements (the LZ formula).

Since the function $|t|^{i\gamma}$ is undefined for $t=0$, we separate the term $(1+t^2)^{i\gamma/2}$ which is equivalent to $|t|^{i\gamma}$ when $t\to\pm\infty$.

Now let us perform the limits $t_0\to-\infty$ and $t\to+\infty$. We use the following asymptotic expansion for the generalized hypergeometric function \cite{Kim}:
\[
{}_2F_2\left(\left.\begin{matrix}
1,a\\b,c
\end{matrix}\right|z\right)\simeq\frac{\Gamma(b)\Gamma(c)}{\Gamma(a)}
[K_{22}(z)+L_{22}(-z)]
\]
as $|z|\to\infty$, $-3\pi/2<\arg z<\pi/2$. Here
$$K_{22}(z)=z^ve^z{}_2F_0(b-a,c-a|z^{-1}),\quad v=1+a-b-c.$$
Series (\ref{EqGenHyperGeom}) diverges whenever $p\geq q+2$ and $z\neq0$, but here function ${}_2F_0$ is understood as just an asymptotic series. Symbol $\simeq$ denotes that the right-hand side is an asymptotic series of the left-hand side function. Further, $L_{22}(-z)$ is defined by formula (3) on p.~15 of~\cite{Luke-Bessel} (the formula given in~\cite{Kim} is valid only for non-integer positive $a$). An algorithm for obtaining this formula is given in~\cite{Luke-book}. We do not write the whole asymptotic series here because of its cumbersomeness. For our purposes, we need only its first term:
\[L_{22}(-z)=\frac{1}{\sqrt\pi z}[\ln z-\psi(\frac12)]+\ldots=\frac{1}{\sqrt\pi z}[\ln 4z+\bm C]+\ldots\]
(the next term has the order $1/z^2$), where $\psi(x)$ is the the digamma function, i.e., logarithmic derivative of the gamma function, $\psi(1/2)=-\bm C-2\ln2$, and $\bm C\approx0.577$ is the Euler (or Euler--Mascheroni) constant. Thus,
\[{}_2F_2\left(\left.\begin{matrix}
1,1\\\frac32,2
\end{matrix}\right|-\frac{it^2}2\right)=
\frac{\ln(2it^2)+\bm C}{it^2}+O(t^{-3}).\]

The function $(t^2/2){}_2F_2(-it^2/2)$ logarithmically diverges when $|t|$ indefinitely increases. But the divergent term is compensated by the term $-\frac i2\ln(1+t^2)/(1+t_0^2)$ in the expressions for $u_{00}(t,t_0)$ and $u_{11}(t,t_0)$. This was the motivation to separate the term $[(1+t^2)/(1+t_0^2)]^{i\gamma/2}$.

Also we use
\[
\lim_{t\to\pm\infty}\mathcal F(t)=\pm\frac{1+i}2=\pm\frac{e^{i\pi/4}}{\sqrt2}.
\]

So,
\begin{widetext}
\begin{subequations}\label{Equsmall}
\begin{eqnarray}
\tilde u_{00}(t,-\infty)=1&-&\frac{\pi\gamma}2\left[1+ \sqrt2e^{i\pi/4}\mathcal F^*(t/\sqrt2)\right]-\frac{\gamma t^2}2\,{}_2F_2(-\frac{it^2}2)-\frac{i\gamma}2\{\ln[2(1+t_0^2)]+\bm C\}+\frac{\pi\gamma}4+O(\gamma^2),\\
\tilde u_{00}(+\infty,t_0)=1&+&\pi\gamma\mathcal F(t_0/\sqrt2)[\frac{e^{-i\pi/4}}{\sqrt2}-\mathcal F^*(t_0/\sqrt2)]+\frac{\gamma t_0^2}2\,{}_2F_2(-\frac{it_0^2}2)+\frac{i\gamma}2\{\ln[2(1+t_0^2)]+\bm C\}-\frac{\pi\gamma}4+O(\gamma^2),\qquad\:\\
\tilde u_{10}(t,-\infty)&=&e^{-i\pi/4}\sqrt{\frac{\pi\gamma}2}[1+\sqrt2e^{-i\pi/4}\mathcal F(t/\sqrt2)]+O(\gamma^{3/2}),\\
\tilde u_{10}(+\infty,t_0)&=&e^{-i\pi/4}\sqrt{\frac{\pi\gamma}2}[1-\sqrt2e^{-i\pi/4}\mathcal F(t_0/\sqrt2)]+O(\gamma^{3/2}),\\
\tilde u_{00}(+\infty,-\infty)&=&1-\pi\gamma+O(\gamma^2),\quad \tilde u_{10}(+\infty,-\infty)=e^{-i\pi/4}\sqrt{2\pi\gamma}+O(\gamma^{3/2}).
\end{eqnarray}
\end{subequations}
\end{widetext}
Form this, we can immediately derive (\ref{EqTargetsmall}). Also, we see that
$$|u_{00}(+\infty,-\infty)|^2=1-2\pi\gamma+O(\gamma^2),$$
so that the first-order approximation of the LZ formula with respect to $\gamma$ is reproduced correctly.

\begin{remark}
Since we need only absolute values of matrix elements of the unitary evolution operator, our trick with  separation of the term $[(1+t^2)/(1+t_0^2)]^{i\gamma/2}$ is not necessary. Indeed, in the first-order approximation, the quantities $|u_{jj}(t,t_0)|^2$ include only the real part of $t^2\,{}_2F_2$, which does not contain divergent terms. However, our trick is essential for problems where not only the transition probability but also the phase is important  (for example, for measurement-assisted Landau-Zener-St\"ukelberg interferometry) or for use of higher-order approximations.
\end{remark}

\section{Approximation of the transition probability for the adiabatic regime}\label{App2}

The aim of this appendix is the derivation of (\ref{EqTargetLarge}).
In the case of large $\gamma$, we use the following asymptotic formula for the parabolic cylinder functions \cite{Crothers,Olver}:
\begin{widetext}
\begin{subequations}\label{EqCroth}
\begin{eqnarray}
D_{i\gamma}(2e^{3i\pi/4}\sqrt\gamma\,t)&=&\exp\left\lbrace
\frac{\pi\gamma}4-\frac{i\gamma}2(1-\ln\gamma)-i\gamma\omega(t)\right\rbrace\sqrt{\frac12\left(1-\frac{|t|}{\sqrt{1+t^2}}\right)}(1+O(\gamma^{-1})),\label{EqCroth1}\\
\sqrt\gamma D_{-1-i\gamma}(-2e^{i\pi/4}\sqrt\gamma\,t)&=&\exp\left\lbrace
\frac{\pi\gamma}4+\frac{i\gamma}2(1-\ln\gamma)-\frac{i\pi}4+
i\gamma\omega(t)\right\rbrace\sqrt{\frac12\left(1+\frac{|t|}{\sqrt{1+t^2}}\right)}(1+O(\gamma^{-1})),
\end{eqnarray}
\end{subequations}
\end{widetext}
where
\[
\omega(t)=\mathrm{sgn}(t)\ln(|t|+\sqrt{1+t^2})+t\sqrt{1+t^2}.
\]
Formally, these asymptotics are valid only for $t\neq0$. However, note that the limits $t\to\pm0$ of the above asymptotic formulas for $D_{i\gamma}(\pm 2e^{3i\pi/4}\sqrt\gamma\,t)$ coincide with each other and with the large $\gamma$ asymptotics of $D_{i\gamma}(0)$. Indeed, in view of (\ref{EqD0}),
\[D_{i\gamma}(0)\sim\frac1{\sqrt2}\exp\left\lbrace\frac{\pi\gamma}4-
\frac{i\gamma}2(1-\ln\gamma)\right\rbrace\]
(here $a\sim b$ means $\lim\limits_{\gamma\to\infty}a/b=1$), which coincides with~(\ref{EqCroth1}) for $t=0$. Similarly, the limits $t\to\pm0$ of the above asymptotic formula for $\sqrt\gamma D_{-1-i\gamma}(\pm2e^{i\pi/4}\sqrt\gamma\,t)$ coincide with each other and with the large $\gamma$ asymptotics of $D_{-1-i\gamma}(0)$:
\[\sqrt\gamma D_{-1-i\gamma}(0)\sim\frac1{\sqrt2}\exp\left\lbrace\frac{\pi\gamma}4+
\frac{i\gamma}2(1-\ln\gamma)-\frac{i\pi}4\right\rbrace.\]
So, formulas (\ref{EqCroth}) can be used for all non-negative $t$.

The substitution of these asymptotics in~(\ref{EqLZGenSol}) gives
\begin{subequations}\label{EqLZlargegamma}
\begin{eqnarray}
a(2\sqrt\gamma\,t)&=&\tilde c_1e^{i\gamma\omega(t)}\sqrt{\frac12\left(1+\frac t{\sqrt{1+t^2}}\right)}+O(\gamma^{-1})\nonumber\\&&+\tilde c_2e^{-i\gamma\omega(t)}\sqrt{\frac12\left(1-\frac t{\sqrt{1+t^2}}\right)},\\
b(2\sqrt\gamma\,t)&=&\tilde c_1e^{i\gamma\omega(t)}\sqrt{\frac12\left(1-\frac t{\sqrt{1+t^2}}\right)}+O(\gamma^{-1})\nonumber\\&&-\tilde c_2e^{-i\gamma\omega(t)}\sqrt{\frac12\left(1+\frac t{\sqrt{1+t^2}}\right)},
\end{eqnarray}
\end{subequations}
where new constants are
\[
\tilde c_{j}=c_{j}\exp\left\lbrace\frac{\pi\gamma}4\pm\frac {i\gamma}2(1-\ln\gamma)+\frac{3i\pi}8\right\rbrace,
\]
Here and in the following, the top  sign in $\pm$ and $\mp$ corresponds to $j=0$ and the bottom sign corresponds to $j=1$. In the following, for simplicity, we will omit the terms $O(\gamma^{-1})$ and solve the problem in the principal approximation with respect to $\gamma$. For the initial conditions $a(2\sqrt\gamma\,t_0)=1$, $b(2\sqrt\gamma\,t_0)=0$, the coefficients are
\begin{equation}\label{EqInitConstA}
\tilde c_j=e^{\mp i\gamma\omega(t_0)}\sqrt{\frac12\left(1\pm\frac{t_0}{\sqrt{1+t_0^2}}\right)}.
\end{equation}
For the initial conditions   $a(2\sqrt\gamma\,t_0)=0$, $b(2\sqrt\gamma\,t_0)=1$,
\begin{equation}\label{EqInitConstB}
\tilde c_j=\pm e^{\mp i\gamma\omega(t_0)}\sqrt{\frac12\left(1\mp \frac{t_0}{\sqrt{1+t_0^2}}\right)}.
\end{equation}

Let us introduce a new time variable
\[
\tau(t)=\frac t{\sqrt{1+t^2}}\in[-1,1]
\]
($\tau=\pm1$ correspond to $t=\pm\infty$). From (\ref{EqLZlargegamma})--(\ref{EqInitConstB}), we derive
\begin{eqnarray*}
&&|u_{jl}(2\sqrt\gamma\,t,2\sqrt\gamma\,t_0)|^2=\frac12(1\pm\tau\tau_0)\\&&\qquad+
\frac12\sqrt{(1-\tau_0^2)(1-\tau^2)}\cos\gamma\Delta\omega t,
\end{eqnarray*}
where $j,l=0,1$, $\tau=\tau(t)$, $\tau_0=\tau(t_0)$, $\Delta\omega_t=\omega_t-\omega_{t_0}$, and the ``plus'' and ``minus'' signs in $\pm$ correspond to $l=j$ and $l=1-j$, accordingly.

Let $2\sqrt\gamma\,t_0$ be the initial time instant (it is equal to $-\infty$ for our problem, but, for a while, we will consider a more general case with an arbitrary $t_0$), with $N$ measurements performed at the time instants $2\sqrt\gamma\,t_1,\ldots,2\sqrt\gamma\,t_N$. Let us define functions (again, not to be confused with the functions $f_n$ in Secs.~\ref{SecLZdynprog} and~\ref{SecSmall})
\begin{eqnarray*}
&&f_j^0(t_N)=|u_{0j}(+\infty,t_N)|^2,\\
&&f_j^k(t_{N-k},t_{N-k+1},\ldots,t_N)=\\&&\qquad\sum_{j_{N-k+1},\ldots,j_N\in\{0,1\}} \prod_{l=N-k}^{N}|u_{j_{l+1},j_l}(t_{l+1},t_l)|^2,
\end{eqnarray*}
$k=1,2,\ldots,N$, where $j_{N-k}=j$, $j_{N+1}=0$, and $t_{N+1}=+\infty$ [cf. (\ref{EqTarget})]. Obviously,
\begin{eqnarray*}
&&f_j^{k+1}(t_{N-k-1},\ldots,t_N)=\\
&&\quad|u_{0j}(t_{N-k},t_{N-k-1})|^2f_0^k(t_{N-k},\ldots,t_N)\\
&&\quad+|u_{1j}(t_{N-k},t_{N-k-1})|^2f_1^k(t_{N-k},\ldots,t_N).
\end{eqnarray*}

By induction, it is not hard to prove that
\begin{widetext}
\begin{eqnarray*}
&&f_j^k(2\sqrt\gamma\,t_{N-k},\ldots,2\sqrt\gamma\,t_N)=\frac12\left\lbrace1\pm\tau_N\prod_{l=N}^{N-k+1}\left[\tau_l\tau_{l-1}+
\sqrt{(1-\tau_l^2)(1-\tau_{l-1}^2)}\cos\gamma\Delta\omega_l\right]\right\rbrace,\\
&&\braket{0|\rho(+\infty)|0}=\lim_{t_0\to-\infty}f_0^N(2\sqrt\gamma\,t_{0},2\sqrt\gamma\,t_{1},\ldots,2\sqrt\gamma\,t_N)=
\frac12\left\lbrace 1-\tau_1\prod_{l=2}^{N}\left[\tau_l\tau_{l-1}+
\sqrt{(1-\tau_l^2)(1-\tau_{l-1}^2)}\cos\gamma\Delta\omega_l\right]\tau_N\right\rbrace,
\end{eqnarray*}
\end{widetext}
where $\prod_{l=N}^{N-k+1}$ is understood as a product in reverse order, $\tau_l=\tau(t_l)$, $\Delta\omega_l=\omega({t_l})-\omega(t_{l-1})$, and the product $\prod_{l=2}^{N}$ is defined to be equal to one if $N=1$ (so, in this case we have $1-\tau_1\tau_N\equiv 1-\tau_1^2$ in the curved brackets).

Since  $\tau_l\in[-1,1]$, we can write $\tau_l=\cos\alpha_l$, $\alpha_l\in[0,\pi]$. This substitution gives exactly (\ref{EqTargetLarge}).

\end{document}